\documentclass[a4paper,UKenglish,cleveref, autoref, thm-restate]{lipics-v2021}

\usepackage{amsmath}
\usepackage{amsthm}
\usepackage{amssymb}
\usepackage{amsfonts}
\usepackage{graphicx}
\usepackage{makecell}
\usepackage{dcolumn}
\usepackage{cite}
\usepackage{mathrsfs}
\usepackage{float}
\usepackage{multirow}
\usepackage{braket}
\usepackage{color}
\usepackage{tcolorbox}
\usepackage[ruled]{algorithm2e}

\title{Playing Mastermind  on quantum computers \footnote{The authors are ordered alphabetically.}}

\author{Lvzhou Li}{Institute of Quantum Computing and Computer Theory, School of Computer and Engineering, Sun Yat-sen University, China}{lilvzh@mail.sysu.edu.cn}{[orcid]}{[funding]}

\author{Jingquan Luo}{Institute of Quantum Computing and Computer Theory, School of Computer and Engineering, Sun Yat-sen University, China}{luojq25@mail2.sysu.edu.cn}{[orcid]}{[funding]}

\author{Yongzhen Xu}{Institute of Quantum Computing and Computer Theory, School of Computer and Engineering, Sun Yat-sen University, China}{xuyzh23@mail2.sysu.edu.cn}{[orcid]}{[funding]}

\authorrunning{Lvzhou Li, Jingquan Luo, Yongzhen Xu } 

\Copyright{Lvzhou Li, Jingquan Luo, Yongzhen Xu} 



\category{} 

\relatedversion{} 

\nolinenumbers 

\begin{document}

\maketitle

\begin{abstract}
From the 1970s up to now,  
the classic two-player game, Mastermind, has attracted plenty of attention, not only from the public as a popular game, but also from the academic community as a scientific issue.
 Mastermind  with $n$ positions and $k$ colors is formally described as follows.  The codemaker privately chooses a secret  $s\in [k]^n$, and the codebreaker  want to determine $s$ in as few queries like $f_s(x)$  as possible to the codemaker. $f_s(x)$ is called a black-peg query if $f_s(x) =B_s(x)$, and  a black-white-peg query if $f_s(x)=\{B_s(x),W_s(x)\}$, where $B_s(x)$ indicates the number of positions where $s$ and $x$ coincide 
 and $ W_s(x)$  indicates the number
of right colors but being in the wrong position.  The complexity of a strategy is  measured by the number of queries used.

In this work we  study  playing Mastermind on quantum computers in both non-adaptive and adaptive  settings, obtaining  efficient quantum algorithms which are all exact (i.e., return the correct result with certainty) and show  huge  quantum speedups. The contributions are as follows.  (i)  Based on the discovery of new structure information, we construct two non-adaptive quantum algorithms  which determine the secret with certainty and  consume $k-1$  and at most $2\lceil \frac{k}{3}  \rceil$ black-peg queries, respectively. 
(ii)  If adaptive strategies are admitted, a more efficient  exact quantum algorithm is obtained, which consumes only $O(\sqrt{k})$  black-peg queries.   Furthermore, we   prove that   any quantum algorithm  requires at least $\Omega(\sqrt{k})$  black-peg queries. (iii) When  black-white-peg queries are allowed, we  propose an adaptive exact quantum algorithm with  $O(\lceil \frac{k}{n}  \rceil + \sqrt{|C_s|})$ queries, where $C_s$ is the set of colors occupied by $s$. This algorithm breaks through the  lower bound $\Omega(\sqrt{k})$  when $n \leq k \leq n^2$.   (iv) Technically,  we develop a three-step framework for designing quantum algorithms for the general string learning problem, which not only  allows  huge quantum speedups  on playing Mastermind, but also may shed light on exploring quantum speedups for  other string learning problems.

Our results show that  quantum computers have  a substantial  speedup  over classical computers on playing Mastermind.  In the non-adaptive setting, when $k \leq n$, the classical complexity is $\Theta(\frac{n \log k }{\max\{\log(n/k), 1\}})$  and  when $k > n$ the current best classical algorithm has  complexity $O(k\log k)$.  In the adaptive setting,  the classical complexity is    $\Theta(n\frac{\log k}{\log n} + \frac{k}{n})$.  Therefore, significant gap between quantum and classical computing on playing Mastermind is clearly visible. 
\end{abstract}

\section{Introduction}

One of the core issues in the field of  quantum computing is to discover more problems that admit quantum speedups, and  to further design explicit quantum algorithms for these problems. Mastermind seems a good candidate for such problems, since the research on Mastermind in classical computing has continued since the 1970s and is intricate, which indicates that the structure of the problem is elusive from the perspective of classical computing, whereas in this paper we will obtain more efficient quantum algorithms with huge quantum speedups.

Mastermind is a classic board  game  invented in 1970 by the Israeli telecommunication expert Mordechai Meirowitz and can go back to the early work of Erd\H{o}s and R\'enyi \cite{erdos1963} in 1963.
As mentioned in \cite{Anders2020},  before the commercial board game version of Mastermind was released in 1971,  variations of this game have been played earlier under other names, such as the
pen-and-paper based games of Bulls and Cows, and Jotto. The game has been played on
TV as a game show in multiple countries under the name of Lingo. Recently, a similar
web-based game has gained much attention under the name of Wordle.

\begin{figure}[htp]
    \centering
    \includegraphics[width=0.35\textwidth]{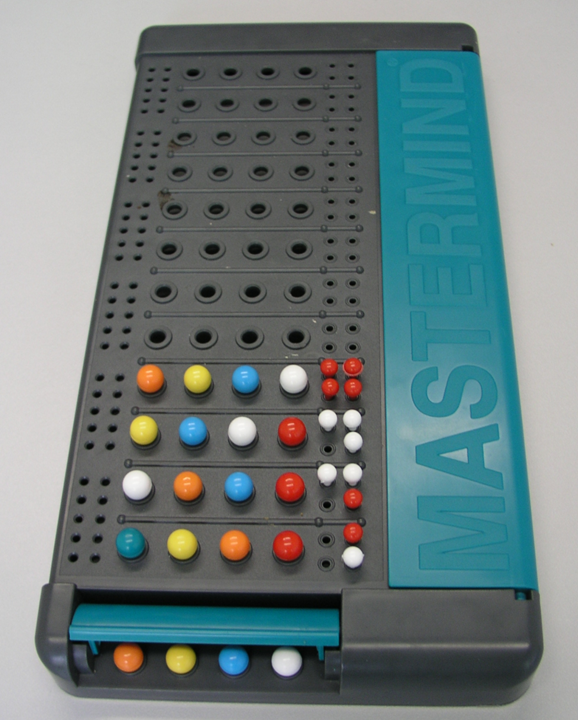}
    \caption{Mastermind game  with four pegs and six colors. The image comes from Ref.\cite{doerr2016} for academic purposes.}
    \label{figMastermind}
\end{figure}

In the commercial version of Mastermind, as shown in Figure \ref{figMastermind}, there are four pegs (positions), each of which could be selected from a set of six colors.
The codemaker  secretly chooses a color combination of four pegs. The goal of the codebreaker  is to identify the secret in as few guesses as possible. In each round, he guesses a color combination of length $4$ to tell the codemaker, and he receives two numbers (i.e., $BW_s(x)$ formally defined later) about how similar the guess is to the secret.  In 1977, Knuth \cite{knuth1976} proved  that $5$ queries are  sufficient for a deterministic algorithm to identify the secret.
In 1983, Chv{\'a}tal \cite{chvatal1983} first studied the generalized version of Mastermind, i.e. $n$ positions and $k$ colors, which will be the topic of this paper.  In the following, we first recall the formal definition.

\subsection{Mastermind}
Let $[k]=\{0,1,\dots k-1\}$ throughout this paper. The Mastermind game with $n$ positions and $k$ colors is formally described as follows. At the start of the game, the codemaker chooses a secret string $s\in [k]^n$. In each round, the codebreaker guesses a string $x\in[k]^n$ and the codemaker replies with $B_s(x)$ or $BW_s(x)$.  The codebreaker ``wins'' the game if he  gets the secret string $s$, with the goal being  to win the game  by using as few queries to $B_s(x)$ or  $BW_s(x)$  as possible.  When we mention ``complexity'' in this paper, it always means the number of queries used by the codebreaker. We call the game \textbf{Black-peg Mastermind} when only the black-peg query is allowed, and \textbf{Black-white-peg Mastermind} otherwise. It's obvious that any algorithm for  Black-peg Mastermind  works for  Black-white-peg Mastermind, but the reverse is not true.
\begin{itemize}
    \item[(a)] \textbf{Black-peg query}:  A black-peg query means an invocation to the  function $B_s$ associated with  $s\in [k]^n$ that returns $B_s(x) =|\{i \in \{1,2,\dots, n\}: s_i = x_i  \}|$ for any $x\in [k]^n$ indicating the number of positions where $s$ and $x$ coincide.  
 \item[(b)] \textbf{Black-white-peg query}:  Similarly, a black-white-peg query means an invocation to the  function $BW_s$ that returns $BW_s(x)=\{B_s(x),W_s(x)\}$  for any $x\in [k]^n$, with $ W_s(x) = \max_{\sigma \in P_n} |\{i \in \{1,2,\dots, n\}: s_i = x_{\sigma(i)}  \}| -B_s(x)$ indicating the number of  right colors but
being in the wrong position, where $P_n$ denotes  the set of all permutations of the set $\{1,2,\dots,n\}$.
        \end{itemize}

Mastermind has sparked a flurry of research, with scholars studying different variations. The components of  variations are as follows \cite{chvatal1983,Berger2018}:
\begin{itemize}
    \item {\bf The color number $k$ and the position number $n$.} For example, the original version considered by Knuth \cite{knuth1976} is with $k=6$ and $n=4$. When $k=2$, the problem  reduces to identifying a binary string, and the problem considered by Erd\H{o}s and R\'enyi \cite{erdos1963} is essentially equivalent to this problem.  The degree to which people understand the complexity of Mastermind depends on the relationship between $k$ and $n$. Thus,  the following cases were usually separately considered in classical computing:   $k=n$,  $k<n^{1-\epsilon}$ with $\epsilon > 0$, and $k>n$. In quantum computing, the distinction among these different cases seems unnecessary.

    \item {\bf The types of query information.} The  game is called Black-peg Mastermind when only the black-peg query is allowed, and Black-white-peg Mastermind otherwise. From the definitions, it can be seen that the black-white-peg query offers more information than the black-peg query, which has  also been strictly verified in  classical computing, since  Black-white-peg Mastermind  has a lower complexity than the corresponding  black-peg version as shown in \cite{doerr2016,Anders2020,Anders2022}.  The similar phenomenon also exists in the quantum situation as will be shown in this paper.

    \item {\bf The query strategy.}
    Depending on the strategy of how an algorithm makes a query, algorithms can be  divided into two kinds: adaptive  and non-adaptive. In the adaptive strategy, the queries  can be made sequentially one by one, and  the next query can depend on the previous queries and the answers. In the non-adaptive strategy, all the query strings must be supplied in parallel at once, and then the secret is determined according to the returned answers without submitting any additional queries. Mastermind in the non-adaptive case was also called static Mastermind \cite{Goddard2003}. Note that the adaptive strategy is  general and the non-adaptive one is more limited.
    \item {\bf Whether  errors are allowed.} Deterministic algorithms output results without errors. Randomized (probabilistic) and quantum algorithms  usually output  results with bounded error. Sometimes, one can  de-randomize a randomized algorithm, obtaining a deterministic one. Also, {\it exact } quantum algorithms  that   output  results with certainty have received much attention. For example,  the  Deutsch-Jozsa algorithm \cite{deutsch1992rapid} and the Bernstein-Vazirani algorithm\cite{Bernstein1997} are all exact. Simon's algorithm can also be improved to be exact \cite{BrassardH97}. In this paper, the quantum algorithms constructed for  Mastermind  are all exact.
    
    \item {\bf Whether repeated colors  are allowed.} {Unless otherwise specified, the Mastermind game we are talking about in this paper allows color repetition, that is,  both $s$ and $x$ are allowed to have the same color in different positions. However, it is worth mentioning that there are some papers considering Mastermind without color repetition \cite{Ouali2018,Glazik2021,Larcher2022}.}
\end{itemize}

From the 1970s up to now,  
Mastermind has attracted a lot of attention, not only from the public as a popular game, but also from the academic community as a scientific issue, especially from the field of  mathematics and computer science (e.g. a partial list of references \cite{erdos1963,knuth1976,chvatal1983,Goodrich2009black,doerr2016,Jiang2019a}). Mastermind has been shown to have a closed relation to  information theory and graph theory \cite{Metric2007,Jiang2019a}.  Also, it has  been used to study other problems, such as being a benchmark problem for intelligent algorithms (genetic and evolutionary algorithms) \cite{Kalisker2003},  understanding the intrinsic difficulty of  heuristics\cite{Droste2006}, simulating pair attacks on genomic data\cite{Goodrich2009a},  and cracking bank ciphers\cite{Focardi2010}.


\subsection{Contributions}

While the classical strategies for playing Mastermind have been studied extensively and deeply, a natural question follows:  Can we construct more efficient quantum algorithms for playing Mastermind? 

We answer the above question affirmatively, obtaining    algorithms with huge quantum speedups for Mastermind in both  non-adaptive and adaptive settings. Our results are summarized in  Table \ref{tableQuantum}  and  the classical  results are presented in Table \ref{tableClassical}.   Note that  the lower bound for the adaptive setting is a trivial lower bound for the non-adaptive setting, and an algorithm with black-peg queries leads to a trivial algorithm with black-white-peg queries.

\begin{table}[htb] 
\setlength{\belowcaptionskip}{10pt}
\begin{center}   
\caption{Results for  Mastermind in the quantum model.}  
\label{tableQuantum} 

\begin{tabular}{|c|c|c|c|}   
\hline   \textbf{} &Adaptive &   Non-adaptive \\  
\hline \makecell{  Black-peg   }  & $\Theta(\sqrt{k})$ [Theorems \ref{theorem:adaptive}, \ref{lowerbound-adaptive-blackpeg}] &  $ O(k) $[Theorem \ref{theorem:Nonadapt_OK}] \\ 
\hline    Black-white-peg  &  $O(\min(\sqrt{k}, \lceil \frac{k}{n}  \rceil + \sqrt{|C_s|}))$ [Theorem \ref{theoremadabwp}] &  $/$ \\ 
\hline   
\end{tabular}   
\end{center}   
\end{table}

\begin{table}[htb] 
\setlength{\belowcaptionskip}{10pt}
\begin{center}   
\caption{Results for  Mastermind in the classical model. Note that ``classical'' here means that the algorithm can be deterministic or randomized.}  
\label{tableClassical} 
\resizebox{\linewidth}{!}{
\begin{tabular}{|c|c|c|c|}   
\hline   \textbf{} & \scriptsize Adaptive & \multicolumn{2}{c|}{\scriptsize Non-adaptive } \\  
\hline   \multirow{2}*{\scriptsize Black-peg }  & \multirow{2}*{ \makecell{\scriptsize $\Theta(n\frac{\log k}{\log n} +k)$  \\ \scriptsize \cite{Anders2020,Anders2022} }}  & \scriptsize $k \leq n $ & \scriptsize $k>n$\\  
\cline{3-4}   ~ & ~ &  \makecell{\scriptsize $\Theta(\frac{n \log k }{\max\{\log(n/k), 1\}})$ \\ \scriptsize \cite{chvatal1983,doerr2016} }  & \makecell{ \scriptsize $\Omega(n\log k) \sim  O(k \log k)$ \\ \scriptsize \cite{doerr2016,Berger2018} }\\ 
\hline   \multirow{2}*{\scriptsize Black-white-peg }  & \multirow{2}*{ \makecell{\scriptsize $\Theta(n\frac{\log k}{\log n} +\frac{k}{n})$  \\ \scriptsize \cite{Anders2020,Anders2022} }}   & \scriptsize $k \leq n^{1-\epsilon} $ for any fixed $\epsilon>0$& \scriptsize others\\  
\cline{3-4}   ~ & ~ &  \makecell{\scriptsize $\Theta(\frac{n \log k}{\log n})$ \scriptsize \cite{chvatal1983} }  & /\\ 
\hline   
\end{tabular}}   
\end{center}   
\end{table}

More specifically, our main results  are as follows. \begin{itemize}
    \item [(i)]{\bf Non-adaptive setting with black-peg query.}  This is the most demanding situation: the query information  is weak, and the restriction on the query strategy is strong.  We obtain two  non-adaptive quantum algorithms   for Black-peg Mastermind  that  return the secret with certainty and consume $k-1$  and at most $2\lceil \frac{k}{3}  \rceil$ black-peg queries, respectively   (Theorem \ref{theorem:Nonadapt_OK}).   It must be pointed out that the two algorithms rely heavily on the discovery of  new structure information of Mastermind, rather than simply applying any existing algorithm to solve the problem.  In addition, it seems not easy to design an efficient non-adaptive quantum algorithm because the non-adaptive characteristic is easy to be destroyed during the algorithmic design process.
 
\item [(ii)] {\bf Adaptive setting with black-peg query.} In this situation, the restriction on the query strategy is removed.
An adaptive quantum algorithm  is constructed for Black-peg Mastermind  that uses $O(\sqrt{k})$ black-peg queries and  succeeds with certainty (Theorem \ref{theorem:adaptive}). Also we prove that any quantum algorithm needs at least $\Omega(\sqrt{k})$ black-peg queries (Theorem \ref{lowerbound-adaptive-blackpeg}).  It is worth pointing out that we need to be very careful when proving lower bounds, since some subtle mistakes are likely to occur. Note that   the  non-adaptive algorithm is more practical than the  adaptive one, since the former needs only to run a shorter quantum circuit  $O(k)$ {\it times}, whereas the latter runs a longer quantum circuit consisting of  $O(\sqrt{k})$ {\it blocks}. 
\item [(iii)]{\bf Adaptive setting  with black-white-peg query.}  In this situation, the query information is strong, and the restriction on the query strategy is removed.
We proposed an adaptive quantum algorithm that returns the secret with certainty and consumes  $O(\lceil \frac{k}{n}  \rceil + \sqrt{|C_s|})$ black-white-peg queries, where $C_s$ is the set of colors occupied by $s$  (Theorem \ref{theoremadabwp}). This algorithm can break through the lower bound $\Omega(\sqrt{k})$ when $n \leq k \leq n^2$. For instance, when $k= n^{\frac{3}{2}}$, we have $O(\lceil \frac{k}{n}  \rceil + \sqrt{|C_s|})=O(\sqrt{n})$, but the lower bound is $\Omega(\sqrt{k})=\Omega(n^{\frac{3}{4}})$.

\end{itemize}

By comparing the results in   Tables \ref{tableQuantum}   and \ref{tableClassical}, one can see that  quantum algorithms always have a substantial speedup advantage over classical counterparts in both  non-adaptive and adaptive settings.  (1) In the non-adaptive setting, our quantum algorithm needs only $O(k)$ black-peg queries that has no relation with $n$. Contrarily,  when $k \leq n$  the classical complexity is $\Theta(\frac{n \log k }{\max\{\log(n/k), 1\}})$   monotonically increasing with respect to $n$, and  when $k > n$ the current best classical algorithm with  $\Omega(k\log k)$-complexity is still worse than our quantum one.  (2) In the adaptive setting, our quantum algorithm needs only $O(\sqrt{k})$ black-peg queries. In contrast, the classical complexity is  $\Theta(n\frac{\log k}{\log n} +\frac{k}{n})$ (or  $\Theta(n\frac{\log k}{\log n} + k)$), which is  $O(\frac{n}{\log n})$ when $n$ is prominent and is $O(k)$ when $k$ is prominent.

\subsection{Techniques}
In this section, we outline the ideas for obtaining the results. 
First note that the Mastermind problem considered here is an instance of the string learning problem: Alice has a secret  string $s$  and Bob wants to identify this secret string
by asking  as few queries as possible to an oracle provided by Alice that answers some piece of information of $s$.  As will be mentioned in the related work section,  a common idea in the literature about quantum algorithms for this problem is as follows: first convert the  provided  oracle  into the inner product oracle, and then apply the Bernstein-Vazirani algorithm \cite{Bernstein1997}.  Following this idea, we can  obtain a quantum algorithm with $O(k\log k)$ black-peg queries for Mastermind, which  has already been superior to the classical ones. Although the algorithm is not readily available and  require some skillful handling, we will not present it here, and one can refer to  Appendix \ref{appendix:B}. 

One of the core issues in the field of quantum computing is   clearly expressed by ``How Much Structure Is Needed for Huge Quantum Speedups?''(the title of  Aaronson' talk at the 28th Solvay Physics Conference \cite{aaronson2022structure}). 
Thus, in this paper we will reveal  new structure information that supports more efficient and even optimal  quantum algorithms for Mastermind in both non-adaptive and adaptive settings. We believe that the discovered structure  information has general implications for string learning problems and may shed light on quantum algorithmic design for other string learning  problems.

\noindent{\bf  Non-adaptive setting with black-peg query.}

Quantum algorithms are constructed for  $k\geq 3$ and $k=2$, respectively, with different ideas.  
Two algorithms (\textbf{Algorithm \ref{algorithm:nonadaptive2}} and \textbf{Algorithm \ref{algorithm:nonadaptive333}}) will be designed  for  $k\geq 3$,  which  are  more inspiring, ingenious, and general. 
Thus, in the following we focus on the case of  $k\geq 3$.  

 In order to  design quantum algorithms for the case of  $k\geq 3$, 
 we develop a three-step framework for designing quantum algorithms for the string learning problem, by discovering  a new structure  that not only  allows huge quantum speedups  on Mastermind in the non-adaptive setting, but also is very likely helpful for exploring quantum speedups on  other string learning problems with different types of query oracles.  More specifically, our framework is as follows:

{\it (a) Discover new structure allowing  quantum speedups.}  For the secret $s\in[k]^n$, we first define the $k\times n$ characteristic matrix $M$  with rows  indexed by colors in $[k]$  and columns indexed by positions in $\{1,2,\dots, n\}$:
\begin{equation*}
    M (c_i, j)= 
        \begin{cases}
            & 1, ~~~s_j = c_i , \\
            & 0, ~~~\text{otherwise}.
        \end{cases}
\end{equation*} 
Then we have
    $s = \sum\limits_{c_i\in [k]} c_i \cdot M (c_i,*),$ 
where $M (c_i,*)$  denotes the $c_i$-th row  of $M$.

A crucial observation is  that  the set $\{M^{(c_2,c_1)}, M^{(c_3,c_1)}, \dots, M^{(c_k,c_1)}\}$ where $M^{(c_l,c_h)}$ denotes the sum of the $c_l$-th and $c_h$-th rows of $M$  suffices to determine the set $\{ M (c_i,*): c_i\in[k]\}$ and thus determine $s$ (see Lemma \ref{lemma:structure}). Also, we observe that  
 $M^{(c_g, c_l)}$ and $M^{(c_l, c_h)}$  with $c_g \neq c_l \neq c_h$ suffice to determine  the three rows of $M$: $M(c_g,*), M(c_l,*)$, and $ M(c_h,*)$, and furthermore, $s$ can be determined by $\lceil \frac{k}{3}  \rceil$ pairs in the form of  $\{M^{(c_g, c_l)}, M^{(c_l, c_h)}\}$  (see Lemma \ref{lemma:structure2}).

{\it (b) Design quantum procedure for learning $M^{(c_l,c_h)}$.}
For any two different colors $c_l, c_{h}\in [k]$,  there is a quantum procedure to learn $M^{(c_l,c_h)}$ by using one query to the following oracle  (see Lemma \ref{lemmaFind2colorposition}):
$$ {B_s^{(c_l, c_h)}} | x \rangle | b \rangle = | x \rangle | b \oplus_{2^m} {B_s^{(c_l, c_h)}}(x) \rangle, $$ 
where  $m$ is required to satisfy $2 ^ m \geq (n+1)$, and ${B_s^{(c_l, c_h)}}: [2]^n\rightarrow \{0,1,\cdots,n\}$ is  defined as 
   $${B_s^{(c_l, c_h)}}(x) = \sum_{i=1}^n (\delta_{x_i 0}\delta_{s_i c_l} + \delta_{x_i 1}\delta_{s_i c_h}),$$
 which indicates how many positions $i$ satisfy $s_i$ takes color $c_l$ when $x_i = 0$ or  $c_h$ when $x_i = 1$.

{\it (c) Construct $B_s^{(c_l, c_h)}$ from $B_s$.} It is proved that 
the oracle $B_s^{(c_l, c_h)}$ can be constructed by using one query to the black-peg oracle $B_s$ (see Lemma \ref{lemmaBsc}).
      
Therefore,  combining Lemmas \ref{lemma:structure}, \ref{lemmaFind2colorposition} and \ref{lemmaBsc} leads to the non-adaptive quantum algorithm  with $k-1$ black-peg queries, and combining Lemmas \ref{lemma:structure2}, \ref{lemmaFind2colorposition} and \ref{lemmaBsc}  results in the one with  at most $2\lceil \frac{k}{3}  \rceil$ queries.

Please note that the first two steps in our framework are independent of the Mastermind problem and apply to various types of string learning problems, with only the last step being for the Mastermind problem. Therefore, 
one may follow this framework to address  other  string learning problems, with the focus being put on the last step, that is, consider how to construct  $B_s^{(c_l, c_h)}$ from the  oracle in hand.

\noindent{\bf Adaptive setting with black-peg query.}

The idea of the $O(\sqrt{k})$-complexity adaptive quantum algorithm (\textbf{Algorithm \ref{algorithm:adaptive}})   is to   apply $n$ Grover searches \cite{Grover1996} synchronously on $n$ positions. Note that the proportions of  target states in the $n$ synchronous Grover searches are all $\frac{1}{k}$. Thus,  we can apply the exact Grover search\cite{ Brassard2002,long2001grover, hoyer2000arbitrary} to make the algorithm error-free. However,  more careful consideration is required on how to implement the general oracle used in the exact Grover search.

Regarding the proof of the  lower bound of complexity, we would like to remind  that this is not at all as simple as imagined. Why emphasize this? Because in previous versions of this paper it was claimed to have obtained the the tight lower bounds for both adaptive and non-adaptive settings, no matter whether black-peg  or black-white-peg queries are used.  Nobody  ever told us an error on the proof, except for thinking that the proof was too simple. Here we present this seemingly correct conclusion and proof  in the box.

\begin{tcolorbox}
Result: For the Mastermind game with $n$ positions and $k$ colors, any non-adaptive quantum algorithm must require $\Omega(k)$   black-peg  or black-white-peg queries, and any adaptive quantum algorithm must require $\Omega(\sqrt{k})$ black-peg  or black-white-peg queries.

Proof: When $n = 1$, one can see that a black-white-peg query is equivalent   to a black-peg query, and   the problem reduces to  the unstructured search problem: searching for one color in $k$ colors, whose non-adaptive and adaptive quantum lower bounds  are well-known to be $\Omega(k)$  \cite{Koiran2010} and $\Omega(\sqrt{k})$ \cite{BennettBBV97}, respectively.

Analysis: This proof seems simple and correct, but unfortunately it is wrong. The  reason is that the lower bound of $n = 1$ cannot simply be considered as a lower bound of the general problem. As an example to refute the above conclusion, we have constructed  an algorithm  with black-white-peg
queries whose complexity can break through the lower bound  $\Omega(\sqrt{k})$.

\end{tcolorbox}

So far, what we can prove is that any quantum algorithm for Mastermind needs at leas
 $\Omega(\sqrt{k})$ black-peg queries, and the idea is as follows.
Denote by $B(k, n)$ the Black-peg Mastermind  with with $n$ positions and $k$ colors, and denote by $Q(k, n)$ the quantum query complexity of $B(k, n)$. We will show that $Q(k, n) \geq Q(k, m)$ if $n \geq m$, which guarantees $Q(k, n) \geq  Q(k, 1)$. On the other hand,   $B(k, 1)$ is actually the unstructured search problem: searching for one color in $k$ colors, whose  quantum lower bound  is well-known to be $\Omega(\sqrt{k})$ \cite{BennettBBV97}. Thus, we have $Q(k, n) \geq  Q(k, 1)=\Omega(\sqrt{k})$. It is worth pointing that this proof idea cannot be applied to Black-white-peg Mastermind, since the black-white-peg oracle does not have the {\it separability} property that the black-peg oracle has. Actually, the lower bound $\Omega(\sqrt{k})$ no longer holds for  Black-white-peg Mastermind. Also,  it seems infeasible to obtain a nontrivial lower
bound for non-adaptive algorithms by using the similar idea, since the proof process will destroy the non-adaptive
characteristics of the algorithm.

\noindent{\bf Adaptive setting with black-white-peg query.}

The $O(\lceil \frac{k}{n}  \rceil + \sqrt{|C_s|})$-complexity quantum algorithm for  Black-white peg Mastermind consists of two steps: (i) first apply the Bernstein-Vazirani algorithm  $\lceil \frac{k}{n}  \rceil$ times to  learn the color set $C_s\subseteq [k]$ occupied by the secret $s$, and (ii) then run the  above $O(\sqrt{k})$-complexity algorithm  to learn $s$ with $k$ replaced by $|C_s|$. Thus the total complexity is $O(\lceil \frac{k}{n}  \rceil + \sqrt{|C_s|})$.

\subsection{Related Work}

\subsubsection{Classical  algorithms   for Mastermind}

Some results about the classical complexity and algorithms for Mastermind are summarized in the following, covering some but not all related to our study. 

\noindent\textbf{Non-adaptive complexity for black-peg Mastermind.}
In 1983, Chv{\'a}tal  \cite{chvatal1983}  first studied   black-peg Mastermind in the non-adaptive setting, proving that  when $k<n^{1-\epsilon}$ with  $\epsilon > 0$,   $(2+\epsilon) n \frac{1+2 \log k}{\log (n/k)}$ queries are sufficient to determine any secret string, matching the information-theoretic lower bound $\frac{n \log k}{\log C_{n+2}^2}$ to a constant factor. Until 30 years later,  Doerr, Doerr, Sp{\"o}hel and Thomas in a breakthrough paper published in Journal of the ACM \cite{doerr2016} (appearing first in SODA 2013) proved that the non-adaptive query complexity  is $\Theta (n \log k / \max\{\log(n/k), 1\})$ for $k \leq n $, which extends Chv{\'a}tal's result. For $k>n$, the  best upper bound is $O(k \log k)$\cite{doerr2016,Berger2018}, and has a gap away from the lower bound $\Omega(n\log k)$\cite{doerr2016}. 

The non-adaptive complexity of black-peg Mastermind is closely related to  two problems. The coin-weighing problem with a spring scale by Shapiro and Fine in 1960 \cite{Shapiro1960} is equivalent to  black-peg Mastermind with two colors. The  minimum number of queries for black-peg Mastermind with $n$ positions and $k$ colors in the non-adaptive setting is equivalent to the metric dimension of the Hamming graph \cite{Metric2007}. From this perspective, Jiang and Polyanskii \cite{Jiang2019a} recently showed that the minimum number of queries is $(2+o(1))n\frac{\log k}{\log n})$ for any constant $k$.

For $k \leq n $, the non-adaptive  complexity of black-white-peg Mastermind is the same as black-peg Mastermind,  since the non-adaptive strategy using only black-peg queries has reached the entropy lower bound.  In contrast, for $k>n$  it is still not clear whether the non-adaptive strategy with  black-white-peg queries can reduce currently the best bound  $O(k\log k)$ achieved by  
the  one with only black-peg queries \cite{doerr2016,Berger2018}.

\noindent\textbf{Adaptive complexity for black-peg Mastermind.} Chv{\'a}tal \cite{chvatal1983} gave a deterministic adaptive algorithm using $2(n\lceil \log n \rceil -2^{\lceil \log n \rceil }+1)$ guesses for $k=n$.
Subsequently, an algorithm with $n\lceil \log n \rceil + \lceil (2-1/k)n\rceil +k$ queries was proposed by Goodrich \cite{Goodrich2009black} for any parameters $n$ and $k$. This was further improved by J\"ager and Peczarski \cite{Jager2011} to $n \lceil \log n \rceil - n + k + 1$ for the case $k > n$ and $n \lceil \log k \rceil + k$ for the case $k \le n$. For $k=n$, it is worth noting that there is a gap  $\log n$ between the upper bound $O(n \log n)$ in the above results and the entropy lower bound $\Omega(n)$. This gap was reduced to $\log \log n$
by Doerr, Doerr, Sp{\"o}hel and Thomas \cite{doerr2016}. They gave the first separation  between the adaptive and non-adaptive strategies in the case of  $k=n$.  Until recently, Martinsson and Su \cite{Anders2020} presented for the first time a randomized algorithm with query complexity $O(n)$, closing the gap with the lower bound $\Omega(n)$, and  proved that the randomized  complexity of black-peg Mastermind is $\Theta(n\frac{\log k}{\log n} +k)$ for any $n$ and $k$  based on the results of \cite{chvatal1983} and \cite{doerr2016}.
In 2022, Martinsson \cite{Anders2022} achieved the same deterministic  complexity utilizing a general query game framework.

\noindent\textbf{Adaptive  complexity for black-white-peg Mastermind.}
An upper bound of adaptive  complexity  of   black-white-peg Mastermind was shown to be  $2n\log k +4n$ for $n\leq k \leq n^2$ and 
$\left \lceil k/n \right \rceil+ 2n\log k +4n$ for $k\geq n$  by Chv{\'a}tal \cite{chvatal1983}. For $k\geq n$, it was  improved to $2n\lceil \log n \rceil +2n +\lceil k/n \rceil +2$ by Chen, Cunha, and Homer \cite{Chen1996}. Also, Doerr, Doerr, Sp{\"o}hel and Thomas\cite{doerr2016} proved that $\Omega(n\log \log n + \frac{k}{n})$  queries are enough to determine any secret for $k\geq n$.
Recently, Refs.\cite{Anders2020,Anders2022} proved that the randomized  and deterministic  complexities are both $\Theta(n\frac{\log k}{\log n} + \frac{k}{n})$ for any $n$ and $k$.

\subsubsection{Quantum algorithms related to Mastermind.}
Although some anonymous reviewers told us that they and their collaborators considered Mastermind's quantum algorithm very early, 
 to the best of our knowledge, surprisingly there has been few publications related to the quantum complexity/algorithms of Mastermind, except  Refs. \cite{Buhrman2008,hunziker2002quantum}.  In fact, Ref. \cite{Buhrman2008} is just an abstract and one cannot verify the correctness of the conclusions. One may be curious why no full paper has been published. More importantly, assuming that the conclusions in Ref. \cite{Buhrman2008} were correct, one can see the obvious gap between the conclusions obtained there and  here.  It was claimed that 
there exist quantum algorithms with $O(\sqrt{k})$ queries for the case $k \leq n$, and with $O(n)$ queries for the case $n \leq k \leq n^2$, but no algorithm was obtained for the general case.
One message  one can learn from the above is that the study of quantum algorithms for Mastermind has attracted attention very early from the academic community, but the problem has not been solved prior to our work.

Ref. \cite{hunziker2002quantum} was not  devoted directly to Mastermind, because neither the word ``Mastermind'' nor the word ``game'' appeared there.
 Hunziker and Meyer \cite{hunziker2002quantum}  considered the problem of identifying a base $k$ string $a$ given an oracle $h_a$ which returns $h_a(x)=dist(x, a)\bmod r$ with $r=\max \{2,6-k\}$, i.e., the Hamming distance between the query string $x$ and the solution $a$ modulo $r$. This problem is  similar to the black-peg Mastermind game but with a slightly different oracle.
For the convenience of readers, we include in Appendix \ref{appendix:A} the algorithm C with $k>4$  proposed by   \cite{hunziker2002quantum}  which succeeds with  probability  $\frac{1}{2} + \epsilon$ when $n < -k\ln(\frac{1}{2} + \epsilon)$.  Hunziker and Meyer \cite{hunziker2002quantum}  claimed that the algorithm  can be  adjusted to an exact version of Grover's algorithm by the methods in  \cite{ long2001grover, hoyer2000arbitrary}, but this is NOT true (On can refer to Appendix \ref{appendix:A} for the reason). Thus, they obtained only an $O(\sqrt{k})$ algorithm for the case of $n < -k\ln(\frac{1}{2} + \epsilon)$, and the complexity $O(\sqrt{k})$ no longer holds for the general $n$.

\subsubsection{Quantum algorithms/complexity for string learning} 
As mentioned before, the Mastermind problem considered here is an instance of the string learning problem: Alice has a secret  string $s$  and Bob wants to identify this secret string
by asking  as few queries as possible to an oracle provided by Alice that answers some piece of information of $s$.  
There have been some quantum algorithms for string learning problems with different query oracles \cite{Bernstein1997,Dam98, CleveIGNTTY12,IwamaNRT12,AmbainisM14,220611221}. A common idea of these quantum algorithms is to first convert the original oracle into the inner product oracle, and then apply the Bernstein-Vazirani algorithm \cite{Bernstein1997},  as shown in the following examples.
\begin{itemize}
    \item \textbf{Inner product query:} Return the inner product of the input string $x \in [2]^n$ and  the secret string $s \in [2]^n$. 
    Any classical algorithm needs $n$ queries, and the Bernstein-Vazirani algorithm \cite{Bernstein1997} can learn $s$ using a single query with certainty.
    \item \textbf{Standard value query:} Return the value $s_i$ of the input $i$ for $s\in [2]^n$.  Also, any classical algorithm needs $n$ queries. Surprisingly, van Dam \cite{Dam98} proposed a bounded-error quantum algorithm using $n/2 + O(\sqrt{n})$ queries based on the Bernstein-Vazirani algorithm \cite{Bernstein1997}. Later, the matching lower bound was proved in Ref.\cite{Farhi1999}.
    \item \textbf{Balanced query:} Compare the weight of any pair of subsets of $s\in [2]^n$. Classically, the query complexity of the  counterfeit coin problem is $\Theta ( d \log (n/d))$, where $d$ is the Hamming weight of $s$. However, there is a quantum algorithm in Ref.\cite{IwamaNRT12 } based on the Bernstein-Vazirani algorithm \cite{Bernstein1997} using only $O(d^{1/4})$ queries for this problem.
     \item \textbf{Group testing query:} Check if there is $1$  in a substring $s_A$ of $s \in [2]^n$ indexed by  the subset $A \subseteq \{1,2,\dots,n\}$. 
The classical query complexity of the group testing problem is $\Theta ( d \log (n/d))$ where the Hamming weight of $s$ is at most $d$. In 2014, Ambainis and Montanaro \cite{AmbainisM14} proposed a quantum algorithm based on the Bernstein-Vazirani algorithm \cite{Bernstein1997} for this problem using $d \log d$ queries on average, and they also gave a lower bound $\Omega(\sqrt{d})$. This gap was closed by Belovs \cite{Belovs15} via the adversary bound.
\end{itemize}

In addition, there are some interesting works considering  quantum algorithms for   gradient estimation\cite{jordan2005fast, gilyen2019optimizing} and for multivariate mean estimation\cite{cornelissen2022near, van2021quantum}, and a crucial step of these algorithms is also to learn a string. Essentially,  these references have the similar idea: (i) the considered problem is to find an approximation  to the given precision of an $n$-dimensional real vector $v$ hidden by an oracle, (ii)  the problem can be reduced to learning a string $s\in [k]^n$  encoding the approximation of $v$, and (iii) then  convert the given oracle into the inner product oracle, and thus apply the extended Bernstein-Vazirani algorithm \cite{mosca1999quantum} to obtain $s$.

As mentioned before, we can also construct the  $O(k\log k)$-complexity algorithm for Mastermind  following the BV-algorithm-based idea (see   Appendix \ref{appendix:B}). But, here we  obtain more efficient quantum algorithms based on new structure information.

\section{Preliminaries}
Some notations and notion used throughout this paper are introduced here, whereas others will be defined when they appear for the first time.  
$C^d$ denotes a $d$-dimensional Hilbert space.  $\oplus_m$ stands for the operation of modulo $m$ addition. $|A|$ is the cardinality of set $A$. For a positive integer $k$, we denote $\{0, 1, 2, \cdots k - 1\}$ by $[k]$.
For a function $F: A\rightarrow [m]$,  we will use the same notation $F$ to denote the quantum implementation for $F$, usually called {\it quantum oracle}, which works as  $F\ket{x}\ket{y}=\ket{x}\ket{y\oplus_m F(x)}$  for $x\in A$ and $y\in [m]$.  $\delta_{ij}$ indicates whether $i$ equals $j$:   $$\delta_{ij} = 
        \begin{cases}
            & 1, ~~~i=j, \\
            & 0, ~~~i\neq j.
        \end{cases} $$

The quantum Fourier transform on a  $d$-dimensional Hilbert space, denoted by $QFT_d $, is defined by
\begin{equation}
    \label{qft}
   QFT_d\ket{l}= \frac{1}{\sqrt{d}} \sum _{j=0}^{d-1} \omega^{lj}\ket{j},
\end{equation}
with $\omega=e^{2\pi i /d}$  and $l\in\{0,1,\cdots, d-1\}$. The inverse $QFT_d$, denoted by $QFT_d^\dagger $, is defined by

\begin{equation}
   QFT_d^\dagger \ket{l}= \frac{1}{\sqrt{d}} \sum _{j=0}^{d-1} \omega^{-lj}\ket{j}.
\end{equation}

\
\section{Non-adaptive Quantum Algorithm}
In this section, we  consider non-adaptive quantum algorithms for Mastermind. 
We  will construct two quantum algorithms using  $k-1$   and  at most $2\lceil \frac{k}{3}  \rceil$ queries, respectively, for $k\geq 3$ in Sec. \ref{sec: non-adaptive-general}, and another algorithm  using one black-peg queries for $k=2$ in Sec. \ref{sec: non-adaptive-special}. 

\subsection{Non-adaptive Quantum Algorithm for $k\geq 3$}\label{sec: non-adaptive-general}
By first converting
the provided oracle into the inner product oracle and then applying the Bernstein-Vazirani
algorithm \cite{Bernstein1997}, we can obtain a quantum algorithm with $O(k\log k)$ black-peg queries for Mastermind. This algorithm has already beaten the classical algorithms. One can refer to Appendix \ref{appendix:B} for more details and note that the algorithm is adaptive. 
Here we will propose more efficient quantum algorithms. Specifically, we develop a  three-step framework for designing quantum algorithms for the  string learning problem, by discovering  a new structure  that not only  allows huge quantum speedups  on Mastermind in the non-adaptive setting, but also may shed light on exploring quantum speedups for  other string learning problems with different types of query oracles.

\subsubsection{New structure allowing  quantum speedups}

Given a string $x \in [k]^n$ and  $c \in [k]$, the product of $x$ and $c$ is defined by
$ x \cdot c= (x_1 \cdot c)(x_2 \cdot c)\cdots (x_n \cdot c).$ 
Given two strings $x, y \in [k]^n$, the sum of $x,y$, denoted by $x+y$,  is a $n$-length string  with the $i$th position being $ (x_i + y_i)\bmod k$ for $i = 1, 2, \cdots n$.

\begin{definition}
    Let $s$ be any secret string from $[k]^n$. $M$ is a $k \times n$  matrix with rows  indexed by colors in $[k]$  and columns indexed by positions in $\{1,2,\dots, n\}$, defined    by
\begin{equation*}
    M (c_i, j)= 
        \begin{cases}
            & 1, ~~~s_j = c_i , \\
            & 0, ~~~\text{otherwise}.
        \end{cases}
\end{equation*}
$M$ is called  the  characteristic matrix associated with $s$. \label{Def:matrix}
\end{definition}
It is obvious that identifying $s$ is equivalent to learning $M$. More specifically, it is easy to see $M$ has the following properties:

 \begin{itemize}
    \item[(a)]  Every column of $M$ contains only one  $1$ and all the other elements are zero. 
 \item[(b)]  The secret string $s$ can be represented as 
\begin{equation}
    s = \sum\limits_{c_i\in[k]} c_i \cdot M (c_i,*), \label{eq:s}
\end{equation}
where $M (c_i,*)$  denotes the $c_i$-th row  of $M$.
        \end{itemize}


Let $M^{(c_i,c_j)}$ be the sum of the $c_i$-th and $c_j$-th rows of $M$ given by \begin{align}M^{(c_i,c_j)} \equiv M (c_i, *) \oplus M (c_j, *)\label{eq:Mcc}\end{align} where $c_i\neq c_j\in [k]$ and $\oplus$ denotes the bitwise XOR.

The next lemma is an important property of the characteristic matrix $M$.
\begin{lemma} Suppose $M$ is the characteristic matrix of  $s \in [k]^n$.  Then we can determine $s$ from 
 the set $\{M^{(c_2,c_1)}, M^{(c_3,c_1)}, \dots, M^{(c_k,c_1)}\}$ where $c_i \neq c_j$ for  $i \neq j$.
\label{lemma:structure}\end{lemma}

\begin{proof} First from Eq. \eqref{eq:Mcc}, we have
 \begin{align}
         M (c_i, *)= M (c_j, *) \oplus M^{(c_i,c_j)}, \label{MCMCMCC}
          \end{align}
    Next, we can show 
 \begin{align}
M (c_i, *) =  M^{(c_i,c_j)} \wedge M^{(c_i,c_l)}, \label{MMMMCC}
 \end{align}
 where $\wedge$ denotes the bitwise AND, and $\vee$ denotes the bitwise OR in the sequel. Actually, we have
    \begin{align*}
    M^{(c_i,c_j)} \wedge M^{(c_i,c_l)} &= (M (c_i,*) \oplus M (c_j,*)  ) \wedge  (M (c_i,*) \oplus M (c_l,*)  ) \\
    &= (M (c_i,*) \vee M (c_j,*)  ) \wedge  (M (c_i,*) \vee M (c_l,*)  ) \\
    &=  M (c_i,*) \vee (M (c_j,*) \wedge M (c_l,*) )\\
      &=  M (c_i,*) \vee \bf{0}\\
    &=M (c_i,*),
\end{align*}
where the second and fourth equations hold from the fact that two different rows of $M$ never have $1$ on the same position.

Then from Eqs. \eqref{MCMCMCC} and \eqref{MMMMCC},  we  have $$M (c_1, *) =  M^{(c_1,c_2)} \wedge M^{(c_1,c_3)},$$ and for $i=2,3,\dots,k$, there is 
\begin{align*}
    M (c_i, *)&= M (c_1, *) \oplus M^{(c_i,c_1)}\\
    &=(M^{(c_1,c_2)} \wedge M^{(c_1,c_3)}) \oplus M^{(c_i,c_1)}.
\end{align*}
It means that we can get $\{M (c_1, *), M (c_2, *), \dots, M (c_k, *)\}$ from $\{M^{(c_2,c_1)}, M^{(c_3,c_1)}, \dots, M^{(c_k,c_1)}\}$, and thus obtain $s$ according to Eq. \eqref{eq:s}.

\end{proof}
\begin{remark} It can be seen that the $i$-th entry of $M^{(c_l, c_h)}$ is 
        $$M^{(c_l, c_h)}_i= 
        \begin{cases}
            & 1, ~~~s_i = c_l ~\text{or}~  c_h, \\
            & 0, ~~~\text{otherwise},
        \end{cases}$$
      for $i\in\{1, 2, \cdots, n\}$,  which thus indicates the positions where $s$ takes the color  $c_l$ or $c_h$.
\end{remark}

Moreover, another  observation is the following lemma.

\begin{lemma} Suppose $M$ is the characteristic matrix of  $s \in [k]^n$. Then $M^{(c_g, c_l)}$ and $M^{(c_l, c_h)}$  with $c_g \neq c_l \neq c_h$ suffice to determine  the three rows of $M$: $M(c_g,*), M(c_l,*)$, and $ M(c_h,*)$. Furthermore, $s$ can be determined from  $\lceil \frac{k}{3}  \rceil$ pairs in the form of  $\{M^{(c_g, c_l)}, M^{(c_l, c_h)}\}$.\label{lemma:structure2}
\end{lemma}

\begin{proof}
Let $Tri=(c_g, c_l,c_h)$ with $c_g \neq c_l \neq c_h$. 
As shown in Table \ref{tableMM}, the value of $s_i$ can be inferred from $M_i^{(c_g,c_l)}$ and $M_i^{(c_l,c_h)}$ for $i=1,2,\dots,n$.  As a result, on can learn in which position $s$ takes the  color  $c_g$ or $c_l$ or $c_h$ from $M^{(c_g, c_l)}$ and $M^{(c_l, c_h)}$  with $c_g \neq c_l \neq c_h$. That is, $M(c_g,*), M(c_l,*)$, and $ M(c_h,*)$ can be determined.
\begin{table}[h]
\setlength{\belowcaptionskip}{10pt}
\caption{The value of $s_i$ based on $M_i^{(c_g,c_l)}$ and $M_i^{(c_l,c_h)}$ for $i=1,2,\dots,n$.}
    \centering
    \begin{tabular}{|c|c|c|}
       \hline   \textbf{$M_i^{(c_g,c_l)}$} & \textbf{$M_i^{(c_l,c_h)}$} & \textbf{$s_i$} \\  
\hline   \textbf{$1$} & \textbf{$1$} & \textbf{$s_i=c_l$} \\
\hline   \textbf{$1$} & \textbf{$0$} & \textbf{$s_i=c_g$} \\
\hline   \textbf{$0$} & \textbf{$1$} & \textbf{$s_i=c_h$} \\
\hline   \textbf{$0$} & \textbf{$0$} & \textbf{$s_i\not\in  Tri $} \\
\hline
    \end{tabular}
    \label{tableMM}
\end{table}

Following the above idea, one can divide the color set $[k]$ into   $\lceil \frac{k}{3}  \rceil$  triples in the form  of $(c_g, c_l,c_h)$ with $c_g \neq c_l \neq c_h$ such that these triples are as disjoint  as possible. For example, let $k=10$. Then we have the following triples: $(0, 1, 2)$, $(3, 4, 5)$, $(6, 7, 8)$, $(7, 8, 9)$, with only the last triple  having an intersection with the  one in front.  Furthermore, we have $\lceil \frac{k}{3}  \rceil$  pairs in the form of   $\{M^{(c_g, c_l)}, M^{(c_l, c_h)}\}$, each of which is associated with a triple $(c_g, c_l,c_h)$ and can be used to  determine $M(c_g,*), M(c_l,*)$, and $ M(c_h,*)$ as shown above. Therefore, All  rows of $M$ can be determined from  the $\lceil \frac{k}{3}  \rceil$  pairs. As a result, $s$ can be determined by Eq. \eqref{eq:s}.

\end{proof}

\begin{remark}
  One will see soon later that Lemma \ref{lemma:structure} leads to the non-adaptive quantum algorithm for Mastermind with $k-1$ black-peg queries, and Lemma \ref{lemma:structure2} results in the one with  at most $2\lceil \frac{k}{3}  \rceil$ queries.
\end{remark}

\subsubsection{Quantum procedure for learning $M^{(c_l,c_h)}$}
Here we  construct a quantum algorithm that takes any two different colors $c_l, c_h \in [k]$ as an input, and returns  $M^{(c_l,c_h)}$ with certainty. 

First, for the secret  string $s\in [k]^n$ and any  two colors $c_l, c_h \in [k]$ with $c_l \neq c_h $,  we define a function ${B_s^{(c_l, c_h)}}: [2]^n\rightarrow \{0,1,\cdots,n\}$  as 
   $${B_s^{(c_l, c_h)}}(x) = \sum_{i=1}^n (\delta_{x_i 0}\delta_{s_i c_l} + \delta_{x_i 1}\delta_{s_i c_h}),$$
 which indicates how many positions $i$ satisfy $s_i$ takes color $c_l$ when $x_i = 0$ or  $c_h$ when $x_i = 1$.  Its quantum oracle works as 
$$ {B_s^{(c_l, c_h)}} | x \rangle | b \rangle = | x \rangle | b \oplus_{2^m} {B_s^{(c_l, c_h)}}(x) \rangle, $$ 
where $m$ is required to satisfy $2 ^ m \geq (n+1).$

Now we obtain the following result.
\begin{lemma}\label{lemmaFind2colorposition}
There is a  quantum algorithm that returns $M^{(c_l, c_h)}$ with certainty    and consumes one query to the   oracle ${B_s^{(c_l, c_h)}}$.
\end{lemma}

\begin{proof}
      We construct explicitly the quantum algorithm called {\it FindTwoColorPosition} (see Algorithm \ref{algorithm:find2color}). It can be intuitively regarded as  $n$ synchronous executions of the Deutsch algorithm \cite{deutsch1985quantum,deutsch1992rapid}.

\begin{algorithm}[htb]    
    \SetKwFunction{FindTwoColorPosition}{FindTwoColorPosition}
    \SetKwInOut{KWProcedure}{Procedure}
    \SetKwInput{Runtime}{Runtime}
    \caption{FindTwoColorPosition}
    \label{algorithm:find2color}
    \LinesNumbered

    \KwIn {A quantum oracle ${B_s^{(c_l, c_h)}}$ with $s \in [k]^n$ and two different colors $c_l, c_h \in [k]$.}
    \KwOut {A string $M^{(c_l, c_h)} \in \{0, 1\}^n$ satisfying 
    $M^{(c_l, c_h)}_i= 
        \begin{cases}
            & 1, ~~~s_i = c_l ~\text{or}~  c_h, \\
            & 0, ~~~\text{otherwise}.
        \end{cases}$}
    \Runtime{One  query to  ${B_s^{(c_l, c_h)}}$. Succeeds with certainty.}
    \KWProcedure{}
    Prepare the initial state $\ket{\Phi_0}=| 0 \rangle ^ {\otimes n} \ket{0}^{\otimes m-1} | 1 \rangle \in (C^2)^{\otimes n} \otimes (C^2)^{\otimes m}$ with $2 ^ m \geq (n+1)$;

    Apply the unitary transformation $H ^{\otimes n} \otimes H ^{\otimes m} $ to $\ket{\Phi_0}$;
    
    Apply the oracle $B_s^{(c_l, c_h)}$;
    
    Apply the unitary transformation $H ^{\otimes n} \otimes H ^{\otimes m}$;
    
    Measure the first $n$ registers in the computational basis.

\end{algorithm}

    At the first step, we prepare the initial state $$\ket{\Phi_0}=| 0 \rangle ^ {\otimes n} \ket{0}^{\otimes m-1} | 1 \rangle \in (C^2)^{\otimes n} \otimes (C^2)^{\otimes m}.$$
    
    At the second step, applying the unitary operator $H ^{\otimes n} \otimes H ^{\otimes m} $ to $\ket{\Phi_0}$, we get
    \begin{equation*}
       \ket{\Phi_1}  = H ^{\otimes n} \otimes H ^{\otimes m} | \Phi_0 \rangle = \frac{1}{\sqrt{2^n}} \sum_{x = 0}^{2^n-1} | x \rangle \otimes \frac{1}{\sqrt{2^m}} \sum_{y = 0}^{2^m-1} (-1)^y | y \rangle.
    \end{equation*}
    
    At the third step, recall that ${B_s^{(c_l, c_h)}}(x) = \sum_{i=1}^n (\delta_{x_i 0}\delta_{s_i c_l} + \delta_{x_i 1}\delta_{s_i c_h})$ and ${B_s^{(c_l, c_h)}} | x \rangle | b \rangle = | x \rangle | b \oplus_{2^m} {B_s^{(c_l, c_h)}}(x) \rangle$ . Then  after applying the oracle $B_s^{(c_l, c_h)}$, we have 
    \begin{align}
      \ket{\Phi_2 }  
        & = B_s^{(c_l, c_h)} | \Phi_1 \rangle \nonumber\\
        & = \frac{1}{\sqrt{2^n}} \sum_{x = 0}^{2^n-1} \left(| x \rangle \otimes \frac{1}{\sqrt{2^m}} \sum_{y = 0}^{2^m-1} (-1)^y | y \oplus_{2^m} {B_s^{(c_l, c_h)}}(x) \rangle\right) \nonumber\\
        & = \frac{1}{\sqrt{2^n}} \sum_{x = 0}^{2^n-1} \left (| x \rangle \otimes \frac{1}{\sqrt{2^m}} \sum_{y = 0}^{2^m-1} (-1)^{y+B_s^{(c_l, c_h)}(x)-B_s^{(c_l, c_h)}(x)} | y \oplus_{2^m} {B_s^{(c_l, c_h)}}(x) \rangle \right) \nonumber\\
         & = \frac{1}{\sqrt{2^n}} \sum_{x = 0}^{2^n-1} \left((-1)^{{-B_s^{(c_l, c_h)}}(x)} | x \rangle \otimes \frac{1}{\sqrt{2^m}} \sum_{y = 0}^{2^m-1} (-1)^{y+B_s^{(c_l, c_h)}(x)} | y \oplus_{2^m} {B_s^{(c_l, c_h)}}(x) \rangle\right)   \label{eq-key1} \\
        & = \frac{1}{\sqrt{2^n}} \sum_{x = 0}^{2^n-1} (-1)^{{B_s^{(c_l, c_h)}}(x)} | x \rangle \otimes \frac{1}{\sqrt{2^m}} \sum_{y' = 0}^{2^m-1} (-1)^{y'} | y' \rangle  \label{eq-key2} \\
        & = \frac{1}{\sqrt{2^n}} \sum_{x = 0}^{2^n-1} (-1)^{\sum_{i=1}^n (\delta_{x_i 0}\delta_{s_i c_l} + \delta_{x_i 1}\delta_{s_i c_h})}| x \rangle \otimes \frac{1}{\sqrt{2^m}} \sum_{y' = 0}^{2^m-1} (-1)^{y'} | y' \rangle \nonumber \\
        \label{eq:bsc2} & = \frac{1}{\sqrt{2^n}} \bigotimes_{i = 1}^{n} [(-1)^{\delta_{s_i c_l}} | 0 \rangle + (-1)^{\delta_{s_i  c_h}} | 1 \rangle] \otimes \frac{1}{\sqrt{2^m}} \sum_{y' = 0}^{2^m-1} (-1)^{y'} | y' \rangle. \nonumber
    \end{align}
Note that in Eq. (\ref{eq-key1}), we have $(-1)^{y+B_s^{(c_l, c_h)}(x)}=(-1)^{y\oplus_{2^m}B_s^{(c_l, c_h)}(x)}$. Then by letting $y'=y\oplus_{2^m}B_s^{(c_l, c_h)}(x)$, we get Eq. (\ref{eq-key2}).
   
   At the fourth step, applying $H ^{\otimes n} \otimes H ^{\otimes m}$ to $\ket{\Phi_2 } $, we get
    \begin{align}
       \ket{\Phi_3 } 
        & = H ^{\otimes n} \otimes H ^{\otimes m} |\Phi_2 \rangle \\
        & = \bigotimes_{i = 1}^{n} (-1)^{\delta_{s_i c_l}} | \delta_{s_i  c_l} \oplus \delta_{s_i c_h} \rangle \otimes \ket{0}^{\otimes m-1} | 1 \rangle \\
        & = \bigotimes_{i = 1}^{n} (-1)^{\delta_{s_i c_l}} | \delta_{s_i  c_1} \vee \delta_{s_i c_h} \rangle \otimes \ket{0}^{\otimes m-1} | 1 \rangle \label{oplus2vee}
    \end{align}
    where  Eq. \eqref{oplus2vee} holds because $\delta_{s_i  c_l}$ and $\delta_{s_i c_h}$ never be both $1$. 
    
    Finally, by measuring the first $n$ registers,  the algorithm outputs with certainty the string $x$ satisfying $x_i = 1$ if $s_i = c_l$ or $s_i = c_h$, and $x_i = 0$ otherwise.

Note that the algorithm uses one query to  ${B_s^{(c_l, c_h)}}$. 
\end{proof}

\subsubsection{Conversion between two oracles}
Currently we are only the last step away from obtaining a quantum algorithm for Mastermind: showing how to construct ${B_s^{(c_l, c_h)}}$ from the black-peg  oracle $B_s$. Now we are going to finish it.

\begin{lemma}\label{lemmaBsc}
Given $s\in [k]^n$, and two colors $c_l, c_h \in [k]$ with $ c_l \neq c_h$, ${B_s^{(c_l, c_h)}}$ can be constructed by using one  black-peg  oracle ${B_s}$.
\end{lemma}
\begin{proof}
Given a secret string $s = s_1s_2 \dots s_n \in [k]^n$ and $x = x_1x_2 \dots x_n \in [2]^n$, we now describe how to compute ${B_s^{(c_l, c_h)}}(x) = \sum_{i=1}^n (\delta_{x_i 0}\delta_{s_i c_l} + \delta_{x_i 1}\delta_{s_i c_h})$ by using  $B_s$.  The quantum circuit diagram implementing  ${B_s^{(c_l, c_h)}}$  is shown in Figure \ref{fig:lemma6}.
\begin{figure}[htbp]  
\centering\includegraphics[width=13cm]{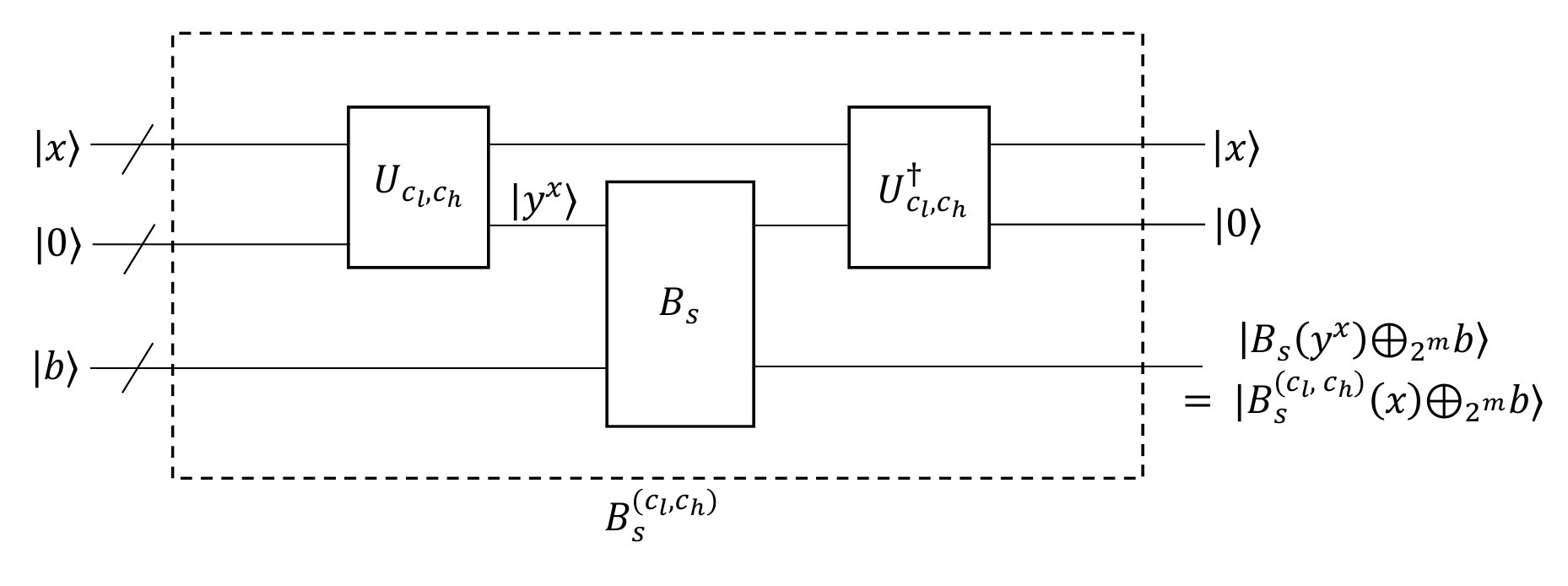} 
\caption{The quantum circuit diagram implementing ${B_s^{(c_l, c_h)}}$  }       
\label{fig:lemma6}   
\end{figure}
First we define a  string $y^x \in [k]^n$ as 
    \begin{equation*}
        y^x_i = 
        \begin{cases}
            & c_l, ~~~if ~   x_i=0, \\
            & c_h, ~~~if~ x_i=1.
        \end{cases}
    \end{equation*}
Let $V = \{ i \in \{1, 2, \dots, n\}: x_i = 0\}$.   
Feeding the black-peg function $B_s$ with $y^x$,  we get
    \begin{align*}
        B_s(y^x) 
        & = |\{ i \in V: s_i = c_l \} \cup \{ i \in \{1, 2, \cdots n \} - V: s_i = c_h \} | \\
        & = \sum_{i=1}^n (\delta_{x_i 0}\delta_{s_i c_l} + \delta_{x_i 1}\delta_{s_i c_h}) \\
        & = {B_s^{(c_l, c_h)}}(x) 
    \end{align*}
As a result,  ${B_s^{(c_l, c_h)}}(x) $ can be computed by calling the black-peg function $B_s$ once.
\end{proof}

\subsubsection{Final quantum algorithm for Mastermind}

Now we have two choices to construct a non-adaptive quantum algorithm for Mastermind.  By combining  Lemmas  \ref{lemma:structure}, \ref{lemmaFind2colorposition} and \ref{lemmaBsc}, we obtain a quantum algorithm with $k-1$ black-peg queries, and  by combining Lemmas  \ref{lemma:structure2}, \ref{lemmaFind2colorposition} and \ref{lemmaBsc}, we obtain the one with  at most $2\lceil \frac{k}{3}  \rceil$ queries.

\begin{theorem}\label{theorem:Nonadapt_OK}
There exist two non-adaptive quantum algorithms for the  Mastermind game with $n$ positions and $ k \geq 3$ colors returning the secret string with certainty:  one uses $k-1$ black-peg queries and the other consumes at most $2\lceil \frac{k}{3}  \rceil$ black-peg queries.
\end{theorem}

\begin{proof}
\begin{algorithm}[htp]
    \SetKwFunction{FindTwoColorPosition}{FindTwoColorPosition}
    \caption{A non-adaptive quantum algorithm for Mastermind  with $n$ positions and $ k \geq 3$ colors using $k-1$ black-peg queries}
    \label{algorithm:nonadaptive2}
    \LinesNumbered
    \SetKwInOut{KWProcedure}{Procedure}
    \SetKwInput{Runtime}{Runtime}
    \KwIn {A black-peg oracle $B_s$ for $s\in [k]^n$ such that $B_s| x \rangle |b\rangle = |x\rangle |b\oplus_{2^m} B_s(x)\rangle$ with $2^m\geq n+1$.}
    \KwOut {The secret string $s$.}
    \Runtime{$k-1$  queries to  $B_s$. Succeeds with certainty.}
    \KWProcedure{}
    \For{i = 2 to k}{
        Construct $B_s^{(c_1, c_i)}$ from $B_s$ and the  pair $( c_1 , c_i )$; ~~~//Lemma \ref{lemmaFind2colorposition} 
        
        Call  \textbf{FindTwoColorPosition}  with $B_s^{(c_1, c_i)}$  as input to get $M^{(c_1, c_i)}$; ~~~//Lemma \ref{lemmaBsc} 
    }
    Set  $M (c_1, *) =  M^{(c_1,c_2)} \wedge M^{(c_1,c_3)}$;  ~~~// Here $ k \geq 3$ is required\;
    \For{j = 2 to k }{
        Set $M (c_j, *)= M (c_1, *) \oplus M^{(c_1,c_j)} $;  ~~~~//Lemma \ref{lemma:structure}
    
    }
    Ouput the secret string $ s = \sum\limits_{i=1}^{k} c_i \cdot M (c_i,*) $. 
\end{algorithm}

\begin{algorithm}[htp]
    \SetKwFunction{FindTwoColorPosition}{FindTwoColorPosition}
    \caption{A non-adaptive quantum algorithm for Mastermind  with $n$ positions and $ k \geq 3$ colors using at most $2\lceil \frac{k}{3}  \rceil$ black-peg queries.} \label{algorithm:nonadaptive333}
    \LinesNumbered
    \SetKwInOut{KWProcedure}{Procedure}
    \SetKwInput{Runtime}{Runtime}
    \KwIn {A black-peg oracle $B_s$ for $s\in [k]^n$ such that $B_s| x \rangle |b\rangle = |x\rangle |b\oplus_{2^m} B_s(x)\rangle$ with $2^m\geq n+1$.}
    \KwOut {The secret string $s$.}
    \Runtime{ At most $2\lceil \frac{k}{3}  \rceil$  queries to  $B_s$. Succeeds with certainty.}
    \KWProcedure{}
Let $T=\{(c_1, c_2, c_3), (c_4, c_5, c_6), \cdots, (c_{k-2}, c_{k-1}, c_k) \}$ where all the triples except the last one are disjoint from each other and last one may has an overlap with the  one in front.
    
    \For{each triple $(c_g, c_l, c_h)\in T$}{

        Construct $B_s^{(c_g, c_l)}$ and  $B_s^{(c_l, c_h)}$from $B_s$ and the  triple $(c_g, c_l, c_h)$; ~~~//Lemma \ref{lemmaFind2colorposition} 
        
        Call  \textbf{FindTwoColorPosition}  with $B_s^{(c_g, c_l)}$  as input to get  $M^{(c_g, c_l)}$; 
        
        Call  \textbf{FindTwoColorPosition}  with $B_s^{(c_l, c_h)}$  as input to get  $M^{(c_l, c_h)}$; ~~~//Lemma \ref{lemmaBsc} 
    
        Get $M(c_g, *), M(c_l, *),  M(c_g, *)$ from  $M^{(c_g, c_l)}$ and $M^{(c_g, c_l)}$;  //Lemma \ref{lemma:structure2}
    
    }
    
    Ouput the secret string $ s = \sum\limits_{i=1}^{k} c_i \cdot M (c_i,*). $ 
\end{algorithm}

  As shown in Algorithm \ref{algorithm:nonadaptive2}, the  idea is to first apply the subroutine \textbf{FindTwoColorPositon} to  get the result  $M^{(c_1, c_i)}$ for $i\in \{2,3,\dots, k\}$ ,  and then obtain all the rows of the characteristic matrix $M$, from which the  secret string $s$ can be deduced.  The correctness of  Algorithm \ref{algorithm:nonadaptive2} is ensured by Lemmas  \ref{lemma:structure}, \ref{lemmaFind2colorposition} and \ref{lemmaBsc}.
 Since  there are $k - 1$ calls to \textbf{FindTwoColorPosition} and each call consumes one  black-peg query to $B_s$ by Lemma \ref{lemmaFind2colorposition},   the complexity of Algorithm \ref{algorithm:nonadaptive2} with respective to $B_s$  is   $k-1$.

 The correctness of  Algorithm \ref{algorithm:nonadaptive333} is ensured by Lemmas  \ref{lemma:structure2}, \ref{lemmaFind2colorposition} and \ref{lemmaBsc}.   \label{algorithm:nonadaptive3}
Note that $|T|=\lceil \frac{k}{3}  \rceil$. Each pair $\{M^{(c_g, c_l)}, M^{(c_l, c_h)}\}$ associated with    $(c_g, c_l, c_h)\in T$  except the last one $(c_{k-2}, c_{k-1}, c_k)$ consumes two queries to $B_s$, and the last one may require only one query, since it is likely to  overlap with the one in front.
Therefore, the algorithm consumes at most  $2\lceil \frac{k}{3}  \rceil$ queries to $B_s$.

\end{proof}

\subsection{Non-adaptive Quantum Algorithm for $k = 2$} \label{sec: non-adaptive-special}
Inspired by the work of \cite{hunziker2002quantum}, we construct a quantum algorithm using only one black-peg query for $k = 2$.
\begin{algorithm}[htb]    
    \SetKwFunction{}{}
    \SetKwInOut{KWProcedure}{Procedure}
    \SetKwInput{Runtime}{Runtime}
    \caption{A non-adaptive quantum algorithm for Mastermind with $n$ positions and $2$ colors}
    \label{algorithm:nonadapt1}
    \LinesNumbered

    \KwIn {A black-peg oracle $B_s$ for $s\in [2]^n$ such that $B_s| x \rangle |b\rangle = |x\rangle |b\oplus_{2^m} B_s(x)\rangle$ with $2^m\geq n+1$.}
    \KwOut {The secret string $s$.}
    \Runtime{One  query to  $B_s$. Succeeds with certainty.}
    \KWProcedure{}
    
    Prepare the initial state $\ket{\Phi_0}=| 0 \rangle ^ {\otimes n} \ket{0}^{\otimes m-2} |0 1 \rangle \in (C^2)^{\otimes n} \otimes (C^2)^{\otimes m}$ with $2 ^ m \geq (n+1)$ and $m\geq 2$;

    Apply the unitary transformation $H ^{\otimes n} \otimes H ^{\otimes m - 2} \otimes QFT_4 $ to $\ket{\Phi_0}$.
    
    Call the black-peg oracle $B_s$ once.
    
    Apply the unitary transformation $ U^{\otimes n} \otimes H ^{\otimes m - 2} \otimes QFT_4^{\dagger}$, where 
    $U = \frac{1}{\sqrt{2}}
    \begin{bmatrix}
    i & 1 \\
    1 & i \\
    \end{bmatrix}.$
    
    Measure the first $n$ registers in the computational basis.

\end{algorithm}

\begin{theorem}
\label{theoremNonk2n}
There is a non-adaptive quantum algorithm for the Mastermind game with $n$ positions and $2$ colors  that uses one black-peg query  and returns the secret string with certainty.
\end{theorem}

\begin{proof}
The non-adaptive algorithm is presented in Algorithm \ref{algorithm:nonadapt1}.

At the first step, prepare the initial state
\begin{align*}
\ket{\Phi_0}=| 0 \rangle ^ {\otimes n} \ket{0}^{\otimes m-2} |0 1 \rangle \in (C^2)^{\otimes n} \otimes (C^2)^{\otimes m}.
\end{align*}

At the second step, apply the unitary operator $H ^{\otimes n} \otimes H ^{\otimes m - 2} \otimes QFT_4 $ to $\ket{\Phi_0}$.  We get
\begin{align*}
    | \Phi_1 \rangle = H ^{\otimes n} \otimes H ^{\otimes m - 2} \otimes QFT_4 | \Phi_0 \rangle = \frac{1}{\sqrt{2^n}} \sum_{x = 0}^{2^n-1}| x \rangle \otimes \frac{1}{\sqrt{2^m}} \sum_{y=0}^{2^m-1} (i)^y | y \rangle.
\end{align*}

At the third step, call the $B_s$ oracle. By noting that $B_s(x) = \sum_{j=1}^n \delta_{s_j x_j} $, we have 
\begin{align}
    | \Phi_2 \rangle 
    & = B_s | \Phi_1 \rangle \nonumber\\
    & = \frac{1}{\sqrt{2^n}} \sum_{x = 0}^{2^n-1} \left(| x \rangle \otimes \frac{1}{\sqrt{2^m}} \sum_{y=0}^{2^m-1} (i)^y | y \oplus_{2^m} B_s(x) \rangle\right) \nonumber\\
    & = \frac{1}{\sqrt{2^n}} \sum_{x = 0}^{2^n-1} \left(| x \rangle \otimes \frac{1}{\sqrt{2^m}} \sum_{y=0}^{2^m-1} (i)^{y + B_s(x) - B_s(x)} | y \oplus_{2^m} B_s(x) \rangle\right) \nonumber\\
    & = \frac{1}{\sqrt{2^n}} \sum_{x = 0}^{2^n-1} \left((i)^{-B_s(x)}| x \rangle \otimes \frac{1}{\sqrt{2^m}} \sum_{y=0}^{2^m-1} (i)^{y + B_s(x)} | y \oplus_{2^m} B_s(x) \rangle\right) \label{eq-key1a}  \\
    & = \frac{1}{\sqrt{2^n}} \sum_{x = 0}^{2^n-1} (-i)^{\sum_{j=1}^n \delta_{s_j x_j}}| x \rangle \otimes \frac{1}{\sqrt{2^m}} \sum_{y=0}^{2^m-1} (i)^{y'} | y' \rangle \label{eq-key2a} \\
    & = \frac{1}{\sqrt{2^n}} \bigotimes_{j = 1}^n ((-i)^{\delta_{s_j 0}} | 0 \rangle + (-i)^{\delta_{s_j 1}} | 1 \rangle ) \otimes \frac{1}{\sqrt{2^m}} \sum_{y=0}^{2^m-1} (i)^{y'} | y' \rangle. \nonumber
\end{align}
Note that in Eq. (\ref{eq-key1a}), we have $(i)^{y+B_s(x)}=(i)^{y\oplus_{2^m}B_s(x)}$\footnote{Let $f(a)=i^a$. Then $4$ is the minimal positive period of $f$. In addition,  it is required in Algorithm \ref{algorithm:nonadapt1} that $m\geq 2$. Thus, it is easy to get  $(i)^{y+B_s(x)}=(i)^{y\oplus_{2^m}B_s(x)}$. }. Then by letting $y'=y\oplus_{2^m}B_s(x)$, we get Eq. (\ref{eq-key2a}).

At the fourth step, after applying the unitary operator $ U^{\otimes n} \otimes H ^{\otimes m - 2} \otimes QFT_4^{\dagger}$, we have
    \begin{align*}
        | \Phi_3 \rangle 
        & = U^{\otimes n} \otimes H ^{\otimes m - 2} \otimes QFT_4^{\dagger} | \Phi_2 \rangle \\
        & = | s_1 s_2 \cdots s_n \rangle \otimes \ket{0}^{\otimes m-1} | 1 \rangle.
    \end{align*}
    
Finally, the secret string $s=s_1s_2\dots s_n$ can be obtained with certainty after measuring the  first $n$ registers in the computational basis.

\end{proof}

\section{Adaptive Quantum Algorithm}

In this section, we discuss  adaptive quantum algorithms for Mastermind. When only  black-peg queries are allowed, we construct a $O(\sqrt{k})$-complexity quantum algorithm in 
Sec. \ref{subsect:adaptive1} and prove the optimality of the algorithm in Sec. \ref{subsect:adaptive2}. When black-white-peg queries are allowed, we present a quantum algorithm with $O(\lceil \frac{k}{n}  \rceil + \sqrt{|C_s|})$ queries in Sec. \ref{subsect:adaptive3}, which is more efficient than the one with black-queries when $n \leq k \leq n^2$ .

\subsection{Adaptive Algorithm with Black-peg Queries}\label{subsect:adaptive1}

Here we will present an  adaptive quantum algorithm  with $\Omega(\sqrt{k})$ black-peg queries.

\begin{theorem}
\label{theorem:adaptive}
There is an adaptive quantum algorithm for the Mastermind game with $n$ positions and $k$ colors that uses $O(\sqrt{k})$ black-peg queries and returns the secret string with certainty.
\end{theorem}

\begin{proof} 
 The adaptive algorithm is presented in Algorithm \ref{algorithm:adaptive}, of which the key idea is to apply $n$  Grover searches  synchronously on $n$ positions. It is well known that  Grover's algorithm  can be adjusted to an exact version  that finds the target state with certainty, if the proportion of the target states, whose value is $\frac{1}{k}$ in our setting, is known in advance.  

\begin{algorithm}[htb]
    \SetKwInput{Runtime}{Runtime}
    \SetKwInOut{KWProcedure}{Procedure}
    \caption{An adaptive quantum algorithm for Mastermind  with $n$ positions and $ k$ colors}
    \label{algorithm:adaptive}
    \LinesNumbered
    \KwIn {A {black-peg} oracle $B_s$ for $s\in [k]^n$ such that $B_s| x \rangle |b\rangle = |x\rangle |b\oplus_{n+1} B_s(x)\rangle$}
    \KwOut {The secret string $s$}
    \Runtime {$O(\sqrt{k})$ queries to $B_s$.  Succeeds with certainty.}
    \KWProcedure{}
    Prepare the initial state $\ket{\Phi_0}=\ket{0}^{\otimes n}\ket{0} \in (C^k)^{\otimes n} \otimes C^{n + 1}$; \label{adaptivestep1} 
    Set the number of  iterations $T = \lceil \frac{\pi}{4 \arcsin(\sqrt{\frac{1}{k}})} - \frac{1}{2} \rceil $ and the rotation angle $\phi = 2 \arcsin(\frac{sin(\frac{\pi}{4T + 2})}{\sin(\theta)})$. \label{adaptivestep2} 
    
    Apply the unitary transformation $QFT_k ^{\otimes n} \otimes I$ to $\ket{\Phi_0}$. 
    
    \For{l = 1 to T} {\label{adaptivestep4}
    
    Apply the unitary operator $O_s(\phi)$, where $O_s(\phi) = B_s^{\dagger}(I \otimes D(\phi)) B_s$, $D(\phi) = \sum_{j = 0}^{n} e^{ij\phi} | j \rangle \langle j |$. 
    
    Apply the unitary operator $S_0(\phi)$, where $S_0(\phi) = ( QFT_k (I + (e^{i\phi} - 1) \ket{0} \langle 0 | ) QFT_k ^{\dagger}) ^{\otimes n} \otimes I $. \label{adaptivestep5} 
    
    } 

    Measure the first $n$ registers in the computational basis.
\end{algorithm}

At the first step, we prepare the initial state $$\Phi_0 = \ket{0}^{\otimes n}\ket{0} \in (C^k)^{\otimes n} \otimes C^{n + 1},$$ where  $(C^k)^{\otimes n}$  is associated with the query registers used to store the query string $x$ and $ C^{n+1}$ is associated with the auxiliary register used to store the query result $B_s(x)$. In addition, we need to set some parameters for  the exact Grover search. There are several approaches to achieve the exact Grover search \cite{ Brassard2002,hoyer2000arbitrary,long2001grover}. Here we use the approach proposed  in \cite{long2001grover}, whose parameters including the number of iterations $T$ and the rotation angle $\phi$ are given below:\footnote{Note that in \cite{long2001grover}, $\phi$ equals $2 \arcsin(\frac{sin(\frac{\pi}{4J + 6})}{\sin(\theta)})$ with the iteration number being $J+1$. If we denote $J'=J+1$, then $\phi = 2 \arcsin(\frac{sin(\frac{\pi}{4J' + 2})}{\sin(\theta)}) $.}
\begin{equation*}
    \begin{split}
        & T = \lceil \frac{\pi}{4 \arcsin(\sqrt{\frac{1}{k}})} - \frac{1}{2} \rceil, \\
        & \phi = 2 \arcsin(\frac{sin(\frac{\pi}{4T + 2})}{\sin(\theta)}) 
           \end{split}
    \label{equation:grover_param}
\end{equation*}
with  $\theta = \arcsin(\sqrt{\frac{1}{k}})$.

At the second step, apply the unitary transformation $QFT_k ^{\otimes n} \otimes I$ to $\ket{\Phi_0}$ to create the uniform superposition state
$$| \Phi_1 \rangle = (QFT_k ^{\otimes n} \otimes I) \ket{\Phi_0} =\frac{1}{\sqrt{k^n}}\sum_{x\in[k]^n}|x \rangle| 0 \rangle=  \bigotimes ^n_{i=1}\left(\frac{1}{\sqrt{k}}\sum_{x_i = 0}^{k - 1}|x_i \rangle\right) \otimes | 0 \rangle.$$ 

From the third  to the sixth step, apply  $T$ Grover iteration operators $\left(S_0(\phi)O_s(\phi)\right)^T$ to $| \Phi_1 \rangle$, where 
\begin{align*}
         S_0(\phi) &= ( QFT_k (I + (e^{i\phi} - 1) \ket{0} \langle 0 | ) QFT_k ^{\dagger}) ^{\otimes n} \otimes I,  \\
        O_s(\phi) &= B_s^{\dagger}(I \otimes D(\phi)) B_s, 
\end{align*}
with
\begin{equation*}
D(\phi) = 
    \begin{bmatrix}
    e^{i 0 \phi} & 0 & 0 & \cdots & 0 \\
    0 & e^{i 1 \phi} & 0 & \cdots & 0 \\
    0 & 0 & e^{i 2 \phi} & \cdots & 0 \\
    \vdots &  & \ddots &  &  \vdots \\
    0 & 0  & \cdots & 0 & e^{i n \phi}\\
    \end{bmatrix}.    
    \label{equation:D}
\end{equation*}

Thus, after the sixth step we get
\begin{align}
| \Phi_2 \rangle&=(S_0(\phi)O_s(\phi))^T| \Phi_1 \rangle\\
&=\bigotimes ^n_{i=1}\left(\left(S'_0(\phi)Q_{s_i}(\phi)\right)^T \frac{1}{\sqrt{k}}\sum_{x_i = 0}^{k - 1}|x_i \rangle\right) \otimes | 0 \rangle \label{pG}\\
&=\ket{s_1s_2\cdots s_n}| 0 \rangle,\label{s_result}
\end{align}
where $ S'_0(\phi) = QFT_k (I + (e^{i\phi} - 1) \ket{0} \langle 0 | ) QFT_k ^{\dagger}$ and $Q_{s_j}(\phi)$ is defined as 
\begin{align}
 Q_{s_j}(\phi)\ket{x_j}= e^{i\phi \delta_{s_jx_j}}\ket{x_j}\label{Qsi}
\end{align}
which is to decide whether $x_j$ equals to $s_j$ or not.  We will explain in more details later why  Eq. (\ref{pG}) holds based on Lemma \ref{lemma:Qs}.  Now assume that it is right. Then one see that  $\left(S'_0(\phi)Q_{s_i}(\phi)\right)^T \frac{1}{\sqrt{k}}\sum_{x_i = 0}^{k - 1}|x_i \rangle$ is actually the exact version of Grover's  algorithm for identifying an $x_i$ such that $x_i=s_i$.  Since the proportions of the target states in $n$ synchronous Grover searches are all $1/k$, the number of the iterations and the rotation angle are the same for each Grover search.
As a result, we get Eq. (\ref{s_result}), and then the algorithm outputs the secret string $s$ with certainty by measuring the first $n$  registers.

The number of iterations of the operator $O_s(\phi)$ is $T = \lceil \frac{\pi}{4 \arcsin(\sqrt{\frac{1}{k}})} - \frac{1}{2} \rceil = O(\sqrt{k})$, and thus  the number of queries to $B_s$ is $O(\sqrt{k})$, which concludes the proof of Theorem \ref{theorem:adaptive}.
\end{proof}
 
Now we are going to explain Eq. (\ref{pG}), which means that the unitary operator $(S_0(\phi)O_s(\phi))^T$  plays a role as  $n$ synchronous Grover searches on $n$ positions. First, $S_0(\phi)$ represents the general diffusion operator of  Grover's algorithm $S'_0(\phi)=QFT_k (I + (e^{i\phi} - 1) \ket{0} \langle 0 | ) QFT_k ^{\dagger}$ applied on $n$ $k$-dimensional spaces in parallel.  Second, we have a look at the effect of $O_s(\phi) = B_s^{\dagger}(I \otimes D(\phi)) B_s$.
Recall that the black-peg oracle $B_s$ works as $B_s| x \rangle |b\rangle = |x\rangle |b\oplus_{n+1} B_s(x)\rangle$, where $|x \rangle \in (C^k)^{\otimes n}$, $| b \rangle \in C^{n + 1}$. 
Then we have 
\begin{lemma}\label{lemma:Qs}
Let $O_s(\phi) = B_s^{\dagger}(I \otimes D(\phi)) B_s$. There is \begin{align*}
       O_s(\phi) |x\rangle |0\rangle= \mathop{\bigotimes}_{j = 1}^{n}  Q_{s_j}(\phi)\ket{x_j}|0\rangle
\end{align*}
for $s=s_1s_2\cdots x_n\in [k]^n, x=x_1x_2\cdots x_n \in [k]^n$.
\end{lemma}
\begin{proof} By direct calculation, we have
\begin{align*}
       O_s(\phi) |x\rangle |0\rangle &=B_s^{\dagger}(I \otimes D(\phi)) B_s |x\rangle |0\rangle \\
        &=B_s^{\dagger}(I \otimes D(\phi))|x\rangle |B_s(x)\rangle \\
        &= e^{i\phi B_s(x)} B_s^{\dagger} |x\rangle|B_s(x) \rangle \\
        &= e^{i\phi B_s(x)} |x \rangle |0 \rangle. 
    \label{equation:2}
\end{align*}
 Note that   $B_s(x) = \sum_{j = 1}^{n} \delta_{s_jx_j}$.
 Thus we have
\begin{align}
       O_s(\phi) |x\rangle|0\rangle  & = e^{i\phi \sum_{j = 1}^{n} \delta_{s_jx_j}} |x_1x_2 \cdots x_n \rangle |0 \rangle \\
       & = \mathop{\bigotimes}_{j = 1}^{n} e^{i\phi \delta_{s_jx_j}} | x_j \rangle  |0 \rangle \label{eq-22}
    \\
   & = \mathop{\bigotimes}_{j = 1}^{n}  Q_{s_j}(\phi)\ket{x_j}|0\rangle, \label{eq-23} 
\end{align}
where  Eq. (\ref{eq-23}) follows from substituting  Eqs. (\ref{Qsi}) into (\ref{eq-22}).
\end{proof}

\subsection{ Tight Lower Bound for Quantum Black-peg Mastermind } \label{subsect:adaptive2}

We first prove the lower bound of quantum complexity for Black-peg Mastermind, and then conclude with some noteworthy remarks.



\begin{theorem}
\label{lowerbound-adaptive-blackpeg}
For the Black-peg Mastermind with $n$ positions and $k$ colors, any  quantum algorithm requires at least $\Omega(\sqrt{k})$ black-peg queries.
\end{theorem}

\begin{proof}  Denote by $B(k, n)$ the Black-peg Mastermind  with with $n$ positions and $k$ colors, and denote by $Q(k, n)$ the quantum query complexity of $B(k, n)$. We will show that $Q(k, n) \geq Q(k, m)$ if $n \geq m$, which leads to $Q(k, n) \geq  Q(k, 1)$. On the other hand,   $B(k, 1)$ is actually the unstructured search problem: searching for one color in $k$ colors, whose  quantum lower bound  is well-known to be $\Omega(\sqrt{k})$ \cite{BennettBBV97}. Thus, we have $Q(k, n) \geq  Q(k, 1)=\Omega(\sqrt{k})$.

It remains to prove $Q(k, n) \geq Q(k, m)$ if $n \geq m$.  It suffices to show that if there is a quantum algorithm for $B(k, n)$, then we can construct a quantum algorithm for $B(k, m)$ with the same query complexity, provided $n \geq m$. We first present the idea in the classical case, and then show that it is feasible in the quantum case. 

Let $s \in [k]^m$ be the secret string of  $B(k, m)$. Firstly, we append a fixed color string of length $n-m$, say $1^{n-m}$, to $s$, obtaining a new secret string $s'$ of length $n$. In the following we show how to implement the black-peg oracle $B_{s'}$ by using the black-peg oracle $B_s$. Let $x$ be any query string to the secret string $s'$.  It is easy to see that 
\begin{equation}
    B_{s'}(x) = B_s(x[1\dots m]) + \sum_{i=m+1}^{n}\delta_{s'_ix_i}.
\end{equation} 
By the same idea, we construct $B_{s'}$ by $B_{s}$ in the quantum case as shown in Figure \ref{fig:lowerbound}. Therefore, if $A$ is an algorithm for $B(k, n)$, then, by replacing the query oracle in $A$ with the circuit in  Figure \ref{fig:lowerbound},  we can obtain an  algorithm  $A'$ for $B(k, m)$ that returns the secret $s'=s1^{n-m}$ where $s$ is the $m$-bit secret we want.

\begin{figure}[htbp]  
\centering\includegraphics[width=10cm]{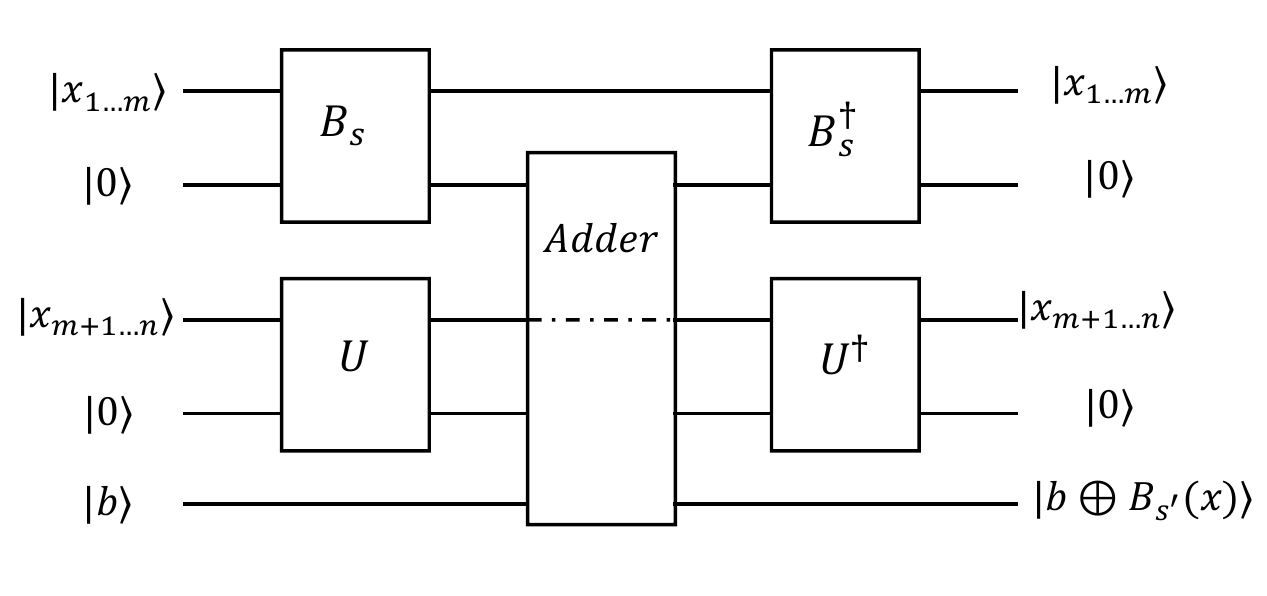} 
\caption{The quantum circuit diagram for implementing $B_{s'}$ with $B_s$ where $U | x_{m+1\dots n} \rangle | 0 \rangle = | x_{m+1\dots n} \rangle | \sum_{i=m+1}^{n}\delta_{s'_ix_i} \rangle $ and $Adder |a\rangle |b\rangle |c\rangle = |a\rangle |b\rangle | c + a + b \rangle$. Note that the wire denoted by  $| x_{m+1\dots n} \rangle$ has no interaction with $Adder$, and thus is depicted with dotted lines when passing through $Adder$. It should be pointed out that the dimension of the auxiliary register of $B_s$ is different from that of $|b\rangle$, so we did not directly add $B_s(x_{1\dots m})$ to $|b\rangle$. Ignoring the two auxiliary registers indicted by $0$, the overall effect the circuit achieves is $|x\rangle |b\rangle \rightarrow |x\rangle | b + B_{s'}(x) \rangle$, as desired.}       
\label{fig:lowerbound}   
\end{figure}
\end{proof}
\begin{remark} 

It is worth mentioning the following two points:
\begin{itemize}
   \item Firstly,  it is not trivial to reduce the problem of small scales to the one of large scales like what we have done in the above proof. In fact, the reason  why the above reduction is available is the  \textbf{separability} property that the  black-peg oracle has, as shown in the following formula:
        \begin{equation}
            B_{s_1 s_2}(x_1 x_2) = B_{s_1}(x_1) + B_{s_2}(x_2),
        \end{equation} 
where $s_i$ and $x_i$ denote a substring.
   However, the  black-white-peg oracle does not satisfy the {separability} property, and thus the above proof does not hold for   black-white-peg Mastermind.  Actually, if black-white-peg queries are allowed, then the quantum  complexity can  break through the lower bound $O(\sqrt{k})$ as shown in Theorem \ref{theoremadabwp}.
    
   \item   Secondly, the reason why we did not use the same method to obtain the lower bound of the non-adaptive quantum complexity for Black-peg Mastermind may not be obvious. In the classical case, the reduction is easy to deal with. However, we must be more careful when addressing the quantum case. As showed in Figure \ref{fig:lowerbound}, we need to call $B_s$  twice sequentially in the conversion process, which  destroys the non-adaptive characteristics of the algorithm. Therefore, it seems infeasible to obtain a lower bound for non-adaptive algorithm by using the idea behind the proof of Theorem \ref{lowerbound-adaptive-blackpeg}.
\end{itemize}
\end{remark}



\subsection{Adaptive Quantum Algorithm with Black-white-peg Queries } \label{subsect:adaptive3}

For the secret $s \in [k]^n$,  let $C_s = \{s_i \in [k]: i\in \{1,2,\dots, n\}\}$, that is, the  set of colors occupied by string $s$. Thus, the size  of $C_s$ is not more than $n$.
For an arbitrary color set $T=\{t_1, t_2,\cdots, t_{|T|}\}\subseteq[k]$  with $|T|\leq n$, a bit string $x^{(T,s)}=x^{(T,s)}_1\dots x^{(T,s)}_{|T|} $ associated with  $s$  is defined by
\begin{equation}
    x^{(T,s)}_i=
    \begin{cases}
        1, &{t_i \in C_s},\\
        0, &{t_i \notin C_s},
    \end{cases}
\end{equation}
which indicates whether the $i$-th color $t_i$ in $T$ is used in $s$ or not. Then we have the following result.

\begin{lemma}\label{lemmabwcs}
  Given the secret $s \in [k]^n$ and   an arbitrary color set $T\subseteq[k]$  with $|T|\leq n$, 
   there is a quantum algorithm that uses $O(1)$ black-white-peg queries and returns $x^{(T,s)}$ with certainty.
\end{lemma}
\begin{proof}
    The idea is as follows: first convert the  provided  oracle  into the inner product oracle, and then apply the Bernstein-Vazirani algorithm \cite{Bernstein1997}. Now we show  the inner product $x^{(T,s)}\cdot y$ for $y \in [2]^{|T|}$ can be computed by using two  black-white-peg queries.
    
First,  we submit the string consisting of   only $1$ to the black-white-peg function $BW_s$ and record the result as $\{B_s(1), W_s(1)\}$. 

Second, given $y = y_1y_2 \dots y_{{|T|}} \in [2]^{|T|}$, we define  a string $z \in [k]^n$ as follows:
    \begin{equation*}
       z_i= 
        \begin{cases}
             t_i, &y_i=1 ~\& ~1 \leq i \leq |T|\\
             1, &y_i=0  ~\& ~ 1 \leq i \leq  |T|\\
             1, & |T| +1 \leq i \leq n.
        \end{cases}
    \end{equation*}
Submit $z$ to the black-white-peg function $BW_s$ and record the result as $\{B_s(z), W_s(z)\}$. Then we have 
\begin{equation*}
    x^{(T,s)}\cdot y=  
    \begin{cases}
         B_s(z) + W_s(z)- \min\{n-|y|, B_s(1)\},  1 \notin \{t_i | y_i = 1, 1 \leq i \leq  |T|\} ~\text{or}~ B_s(1)=0 \\
             B_s(z) + W_s(z)- \min\{n-|y|, B_s(1)-1\}, \text{otherwise}.
    \end{cases}
\end{equation*}

As a result,  $x^{(T,s)}\cdot y$ can be computed by using two black-white-peg queries. Thus,   $x^{(T,s)}$ can be learn with certainty using $O(1)$ black-white-peg queries by the Bernstein-Vazirani algorithm \cite{Bernstein1997}.
\end{proof}

Based on the above result, we obtain a quantum algorithm for identifying the secret $s$.
\begin{theorem}\label{theoremadabwp}
There is an adaptive quantum algorithm for the Mastermind game with $n$ positions and $k$ colors that uses $O(\lceil \frac{k}{n}  \rceil + \sqrt{|C_s|})$ black-white-peg queries and returns the secret $s$ with certainty, where $C_s$ is the set of colors occupied by $s$.
\end{theorem}

\begin{proof}
The color set $[k]$ is divided  into disjoint sets  $T_1, \dots, T_{ \lceil \frac{k}{n}  \rceil }$ such that $|T_i|=n$ for $i<\lceil \frac{k}{n}\rceil $ and $|T_{ \lceil \frac{k}{n}  \rceil }|\leq n$. By Lemma \ref{lemmabwcs}, we can learn $C_s$ with certainty using $O(\lceil \frac{k}{n}  \rceil)$ black-white-peg queries. Now, the problem is to solve Mastermind game with $n$ positions and $|C_s|$ colors. This can be done  using $O(\sqrt{|C_s|})$ black-white-peg queries by Theorem \ref{theorem:adaptive} (the white-peg information is simply ignored).
The overall complexity is  $O(\lceil \frac{k}{n}  \rceil + \sqrt{|C_s|})$.
\end{proof}

\begin{remark}
   When $n \leq k \leq n^2$, the algorithm in Theorem \ref{theoremadabwp} has a complexity lower than  the bound $\Omega(\sqrt{k})$ given in Theorem \ref{lowerbound-adaptive-blackpeg}. For instance, when $k= n^{\frac{3}{2}}$, we have $O(\lceil \frac{k}{n}  \rceil + \sqrt{|C_s|})=O(\sqrt{n})$, but $\Omega(\sqrt{k})=\Omega(n^{\frac{3}{4}})$.
\end{remark}

\section{Conclusions and Discussions}
In  this paper, we have  investigated quantum  algorithms for playing the popular game of  Mastermind, obtaining substantial quantum speedups.    Technically,  we have developed a  framework for designing quantum algorithms for the general string learning problem, by discovering  a new structure  that not only  allows  huge quantum speedups  on plying Mastermind, but also is very likely helpful for addressing  other string learning problems with different kinds of query oracles. It is worth pointing out that the non-adaptive algorithm is more practical than the adaptive one, since the former needs only to run a shorter quantum circuit  $O(k)$ {\it times}, whereas the latter runs a longer quantum circuit consisting of  $O(\sqrt{k})$ {\it blocks}.


In the following we list some problems maybe worthy of further consideration.

\noindent\textbf{ Problem 1:  What is the tight lower bound of the non-adaptive quantum complexity for Black-peg Mastermind?}  We have presented two $O(k)$-complexity non-adaptive
quantum  algorithms for  Black-peg Mastermind. However, it is not clear whether the algorithms are optimal   in the non-adaptive setting.

\noindent\textbf{ Problem 2:  What is the tight lower bound of the quantum complexity for Black-white-peg Mastermind?} For Black-white-peg Mastermind, we have obtained a quantum algorithm with $O(\lceil \frac{k}{n}  \rceil + \sqrt{|C_s|})$ queries. In further work, it is worth exploring the tight lower bound for Black-white-peg Mastermind in both adaptive and non-adaptive settings.

\noindent\textbf{ Problem 3:  What is the quantum  complexity of Mastermind without color repetition?} There is a variation of Mastermind  where color repetition is prohibited in both the secret string $s$ and the query string $x$. In particular, when $k=n$, this variation is called {\it Permutation} Mastermind.  Similar to the case with color repetition, the classical  complexity of Permutation Mastermind in the non-adaptive setting is also $\Theta(n\log n)$ \cite{Glazik2021,Larcher2022}.
The classical  complexity of Permutation Mastermind in the adaptive setting leaves an $O(\log n)$ gap between the lower bound $\Omega (n)$ and the upper bound $O(n \log n)$ \cite{Ouali2018}.
Our quantum algorithms seem not  suitable for this variation.


\bibliography{lipics-v2021-sample-article}

\begin{thebibliography}{10}

\bibitem{aaronson2022structure}
Scott Aaronson.
\newblock How much structure is needed for huge quantum speedups?
\newblock {\em arXiv:2209.06930}, 2022.
\newblock \href {http://arxiv.org/abs/2209.06930} {\path{arXiv:2209.06930}}.

\bibitem{AmbainisM14}
Andris Ambainis and Ashley Montanaro.
\newblock Quantum algorithms for search with wildcards and combinatorial group
  testing.
\newblock {\em Quantum Inf. Comput.}, 14(5-6):439--453, 2014.

\bibitem{Belovs15}
Aleksandrs Belovs.
\newblock Quantum algorithms for learning symmetric juntas via the adversary
  bound.
\newblock {\em Comput. Complex.}, 24(2):255--293, 2015.

\bibitem{BennettBBV97}
Charles~H. Bennett, Ethan Bernstein, Gilles Brassard, and Umesh~V. Vazirani.
\newblock Strengths and weaknesses of quantum computing.
\newblock {\em {SIAM} J. Comput.}, 26(5):1510--1523, 1997.

\bibitem{Berger2018}
Aaron Berger, Christopher Chute, and Matthew Stone.
\newblock Query complexity of mastermind variants.
\newblock {\em Discret. Math.}, 341(3):665--671, 2018.

\bibitem{Bernstein1997}
Ethan Bernstein and Umesh~V. Vazirani.
\newblock Quantum complexity theory.
\newblock {\em {SIAM} J. Comput.}, 26(5):1411--1473, 1997.

\bibitem{BrassardH97}
Gilles Brassard and Peter H{\o}yer.
\newblock An exact quantum polynomial-time algorithm for simon's problem.
\newblock In {\em Proceedings of the Fifth Israel Symposium on Theory of
  Computing and Systems}, pages 12--23, 1997.

\bibitem{Brassard2002}
Gilles Brassard, Peter H{\o}yer, Michele Mosca, and Alain Tapp.
\newblock Quantum amplitude amplification and estimation.
\newblock {\em Contemporary Mathematics}, 305:53--74, 2002.

\bibitem{Buhrman2008}
Harry Buhrman and Andr{\'e} Souto.
\newblock Quantum mastermind.
\newblock 2008.
\newblock This abstract can be found at
  \url{https://www.math.tecnico.ulisboa.pt/~dil2008/index_files/quantum_mastermind_abstract_dil2008.pdf
  }.

\bibitem{Metric2007}
Jos{\'{e}} C{\'{a}}ceres, M.~Carmen Hernando, Merc{\`{e}} Mora, Ignacio~M.
  Pelayo, Mar{\'{\i}}a~Luz Puertas, Carlos Seara, and David~R. Wood.
\newblock On the metric dimension of cartesian products of graphs.
\newblock {\em {SIAM} J. Discret. Math.}, 21(2):423--441, 2007.

\bibitem{Chen1996}
Zhixiang Chen, Carlos Cunha, and Steven Homer.
\newblock Finding a hidden code by asking questions.
\newblock In {\em Proceedings of the Second Annual International Conference on
  Computing and Combinatorics}, pages 50--55, 1996.

\bibitem{chvatal1983}
V.~Chv{\'a}tal.
\newblock Mastermind.
\newblock {\em Combinatorica}, 3(3-4):325--329, 1983.

\bibitem{CleveIGNTTY12}
Richard Cleve, Kazuo Iwama, Fran{\c{c}}ois~Le Gall, Harumichi Nishimura,
  Seiichiro Tani, Junichi Teruyama, and Shigeru Yamashita.
\newblock Reconstructing strings from substrings with quantum queries.
\newblock In {\em Proceedings of the 13th Scandinavian Symposium and Workshops
  on Algorithm Theory}, pages 388--397, 2012.

\bibitem{cornelissen2022near}
Arjan Cornelissen, Yassine Hamoudi, and Sofiene Jerbi.
\newblock Near-optimal quantum algorithms for multivariate mean estimation.
\newblock In {\em Proceedings of the 54th Annual ACM SIGACT Symposium on Theory
  of Computing}, pages 33--43, 2022.

\bibitem{deutsch1985quantum}
David Deutsch.
\newblock Quantum theory, the church--turing principle and the universal
  quantum computer.
\newblock {\em Proc. R. Soc. Lond. A}, 400:97--117, 1985.

\bibitem{deutsch1992rapid}
David Deutsch and Richard Jozsa.
\newblock Rapid solution of problems by quantum computation.
\newblock {\em Proc. R. Soc. Lond. A}, 439:553--558, 1992.

\bibitem{doerr2016}
Benjamin Doerr, Carola Doerr, Reto Sp{\"o}hel, and Henning Thomas.
\newblock Playing mastermind with many colors.
\newblock {\em Journal of the ACM}, 63(5):1--23, 2016.
\newblock Earlier version in SODA 2013, pages 695-704.

\bibitem{Droste2006}
Stefan Droste, Thomas Jansen, and Ingo Wegener.
\newblock Upper and lower bounds for randomized search heuristics in black-box
  optimization.
\newblock {\em Theory Comput. Syst.}, 39(4):525--544, 2006.

\bibitem{erdos1963}
Paul Erd{\H{o}}s and Alfr{\'e}d R{\'e}nyi.
\newblock On two problems of information theory.
\newblock {\em Magyar Tud. Akad. Mat. Kutat{\'o} Int. K{\"o}zl}, 8:229--243,
  1963.

\bibitem{Farhi1999}
Edward Farhi, Jeffrey Goldstone, Sam Gutmann, and Michael Sipser.
\newblock Bound on the number of functions that can be distinguished with
  $\mathit{k}$ quantum queries.
\newblock {\em Phys. Rev. A}, 60:4331--4333, 1999.

\bibitem{Focardi2010}
Riccardo Focardi and Flaminia~L. Luccio.
\newblock Cracking bank pins by playing mastermind.
\newblock In {\em Proceedings of the 5th International Conference on Fun with
  Algorithms}, pages 202--213, 2010.

\bibitem{gilyen2019optimizing}
Andr{\'a}s Gily{\'e}n, Srinivasan Arunachalam, and Nathan Wiebe.
\newblock Optimizing quantum optimization algorithms via faster quantum
  gradient computation.
\newblock In {\em Proceedings of the Thirtieth Annual ACM-SIAM Symposium on
  Discrete Algorithms}, pages 1425--1444. SIAM, 2019.

\bibitem{Glazik2021}
Christian Glazik, Gerold J{\"{a}}ger, Jan Schiemann, and Anand Srivastav.
\newblock Bounds for the static permutation mastermind game.
\newblock {\em Discret. Math.}, 344(3):112253, 2021.

\bibitem{Goddard2003}
Wayne Goddard.
\newblock Static mastermind.
\newblock {\em Journal of Combinatorial Mathematics and Combinatorial
  Computing}, 47:225--236, 2003.

\bibitem{Goodrich2009a}
Michael~T. Goodrich.
\newblock The mastermind attack on genomic data.
\newblock In {\em Proceedings of 30th IEEE Symposium on Security and Privacy},
  pages 204--218, 2009.

\bibitem{Goodrich2009black}
Michael~T. Goodrich.
\newblock On the algorithmic complexity of the mastermind game with black-peg
  results.
\newblock {\em Inf. Process. Lett.}, 109(13):675--678, 2009.

\bibitem{Grover1996}
Lov~K Grover.
\newblock A fast quantum mechanical algorithm for database search.
\newblock In {\em Proceedings of the twenty-eighth annual ACM symposium on
  Theory of computing}, pages 212--219, 1996.

\bibitem{hoyer2000arbitrary}
Peter H{\o}yer.
\newblock Arbitrary phases in quantum amplitude amplification.
\newblock {\em Physical Review A}, 62(5):052304, 2000.

\bibitem{hunziker2002quantum}
Markus Hunziker and David~A Meyer.
\newblock Quantum algorithms for highly structured search problems.
\newblock {\em Quantum Information Processing}, 1(3):145--154, 2002.

\bibitem{IwamaNRT12}
Kazuo Iwama, Harumichi Nishimura, Rudy Raymond, and Junichi Teruyama.
\newblock Quantum counterfeit coin problems.
\newblock {\em Theor. Comput. Sci.}, 456:51--64, 2012.

\bibitem{Jager2011}
Gerold J{\"{a}}ger and Marcin Peczarski.
\newblock The number of pessimistic guesses in generalized black-peg
  mastermind.
\newblock {\em Inf. Process. Lett.}, 111(19):933--940, 2011.

\bibitem{Jiang2019a}
Zilin Jiang and Nikita Polyanskii.
\newblock On the metric dimension of cartesian powers of a graph.
\newblock {\em J. Comb. Theory, Ser. A}, 165:1--14, 2019.
\newblock Earlier version in SODA 2019, pages 1215--1220.

\bibitem{jordan2005fast}
Stephen~P Jordan.
\newblock Fast quantum algorithm for numerical gradient estimation.
\newblock {\em Physical review letters}, 95(5):050501, 2005.

\bibitem{Kalisker2003}
Tom Kalisker and Doug Camens.
\newblock Solving mastermind using genetic algorithms.
\newblock In {\em Proceedings of the Conference on Genetic and Evolutionary
  Computation (GECCO’03)}, pages 1590--1591, 2003.

\bibitem{knuth1976}
Donald~E Knuth.
\newblock The computer as master mind.
\newblock {\em Journal of Recreational Mathematics}, 9(1):1--6, 1977.

\bibitem{Koiran2010}
Pascal Koiran, J{\"{u}}rgen Landes, Natacha Portier, and Penghui Yao.
\newblock Adversary lower bounds for nonadaptive quantum algorithms.
\newblock {\em J. Comput. Syst. Sci.}, 76(5):347--355, 2010.

\bibitem{Larcher2022}
Maxime Larcher, Anders Martinsson, and Angelika Steger.
\newblock Solving static permutation mastermind using $o(n\log n)$ queries.
\newblock {\em Electron. J. Comb.}, 29(1), 2022.

\bibitem{long2001grover}
Gui-Lu Long.
\newblock Grover algorithm with zero theoretical failure rate.
\newblock {\em Physical Review A}, 64(2):022307, 2001.

\bibitem{Anders2022}
Anders Martinsson.
\newblock Optimal schemes for combinatorial query problems with integer
  feedback.
\newblock {\em arXiv:2203.09496}, 2022.

\bibitem{Anders2020}
Anders Martinsson and Pascal Su.
\newblock Mastermind with a linear number of queries.
\newblock {\em arXiv:2011.05921}, 2020.

\bibitem{mosca1999quantum}
Michele Mosca.
\newblock {\em Quantum computer algorithms}.
\newblock PhD thesis, University of Oxford. 1999., 1999.

\bibitem{Ouali2018}
Mourad~El Ouali, Christian Glazik, Volkmar Sauerland, and Anand Srivastav.
\newblock On the query complexity of black-peg ab-mastermind.
\newblock {\em Games}, 9(1):2, 2018.

\bibitem{Shapiro1960}
H.~S. Shapiro and N.~J. Fine.
\newblock E1399.
\newblock {\em The American Mathematical Monthly}, 67(7):697--698, 1960.

\bibitem{van2021quantum}
Joran van Apeldoorn.
\newblock Quantum probability oracles \& multidimensional amplitude estimation.
\newblock In {\em 16th Conference on the Theory of Quantum Computation,
  Communication and Cryptography (TQC 2021)}. Schloss Dagstuhl-Leibniz-Zentrum
  f{\"u}r Informatik, 2021.

\bibitem{Dam98}
Wim van Dam.
\newblock Quantum oracle interrogation: Getting all information for almost half
  the price.
\newblock In {\em Proceedings of the 39th Annual Symposium on Foundations of
  Computer Science}, pages 362--367, 1998.

\bibitem{220611221}
Yongzhen Xu, Shihao Zhang, and Lvzhou Li.
\newblock Quantum algorithm for learning secret strings and its experimental
  demonstration.
\newblock {\em Physica A: Statistical Mechanics and its Applications},
  609:128372, 2023.

\end{thebibliography}

\newpage
\appendix
\section{ The algorithm proposed by  Hunziker and Meyer}
\label{appendix:A}
For the convenience of readers, here we describe the algorithm proposed by  Hunziker and Meyer \cite{hunziker2002quantum}.

The problem considered by Hunziker and Meyer \cite{hunziker2002quantum} is similar to the problem considered in the article, which is aimed to identify an element of $H_k^n$ defined as followed:
\begin{align}
    H_k^n = \{h_a : \{0, \cdots k -1\}^n \rightarrow \{0, 1\} | a \in \{0, \cdots k - 1\}^n ~ and ~ h_a(x) = dist(x, a) \bmod 2 \}
\end{align}
where $dist(x, a)$ is the generalized Hamming distance between $a$ and $x$, i.e., the number of components at which they differ.

\begin{theorem}\cite{hunziker2002quantum}
Algorithm \ref{algorithm:meyer} identifies an element of $H_k^n$ with probability at least $\frac{1}{2} + \epsilon(0 < \epsilon \leq \frac{1}{2})$ for $n \leq -k\ln(\frac{1}{2} + \epsilon)$, using $\lfloor \frac{\pi}{4} \sqrt{k} \rceil$ quantum queries.
\end{theorem}

\begin{algorithm}[htb]
    \SetKwInput{Runtime}{Runtime}
     \SetKwInOut{KWProcedure}{Procedure}
    \caption{Alaogirhm C in \cite{hunziker2002quantum} }
    \label{algorithm:meyer}
    \LinesNumbered
    \KwIn {An oracle $O_{h_a}$ for $a\in [k]^n$ such that $O_{h_a} | x \rangle | b \rangle = | x \rangle | b \oplus h_a(x) \rangle$.}
    \KwOut {The secret string $s$.}
    \Runtime {$O(\sqrt{k})$ queries to $O_{h_a}$.  Succeeds with probability at least $\frac{1}{2} + \epsilon $ when $n < -k\ln(\frac{1}{2} + \epsilon)$.}
    
    \KWProcedure{}
    Initial the state to $ | 0 \rangle ^{\otimes n} | 0 \rangle \in (C_k)^{\otimes n}\otimes C_2$.
    
    Apply the unitary transformation $QFT_{k}^{\otimes n} \otimes (HX)$
    
    \For{i = 1 : $\lfloor \frac{1}{2}(\pi / (2\arcsin(\frac{1}{\sqrt{k}})) - 1)\rceil $}{
        apply the oracle $O_{h_a}$.
        
        apply the unitary transformation $(QFT_k (I - 2|0 \rangle \langle 0 |)QFT_k^{\dagger})^{\otimes n } \otimes I$
    }
    Measure the first $n$ registers.
\end{algorithm}

  Hunziker and Meyer \cite{hunziker2002quantum}  claimed that the algorithm  can be  adjusted to an exact version of Grover's algorithm by the methods in  \cite{ long2001grover, hoyer2000arbitrary}, but this is NOT true  as  explained below.

Note that when $k>4$, we have $h_{a}(x) = ({n - \sum_{i = 1}^n \delta_{s_ix_i}}) \bmod 2$, and the quantum oracle works as  $O_{h_a} | x \rangle | b \rangle = | x \rangle | b \oplus h_a(x) \rangle $ where $\oplus$ denotes  XOR. We explain in  details how the $O_{h_a}$ oracle works in the algorithm as shown in Eqs. \eqref{equation:meyer1} $\sim$ \eqref{equation:meyer4}.  It should be pointed out that  Eq. \eqref{equation:meyer3.5} holds as  $
(-1)^{ l} = (-1)^{ l\bmod 2}$ for any $0\leq l\leq n$, but it will not hold if we replace $-1$ with $ e^{i\phi}$ for general $\phi$, since $
    e^{i\phi l} = e^{i\phi l\bmod 2}$ no longer holds.
However, in the exact Grover search  \cite{ long2001grover, hoyer2000arbitrary} it is necessary to realize a general phase $e^{i\phi}$. That is why the algorithm  given by \cite{hunziker2002quantum} can't be adapted to be exact  by the methods in  \cite{ long2001grover, hoyer2000arbitrary}.
\begin{align}
    O_{h_a} (|x \rangle \otimes \frac{1}{\sqrt{2}}(| 0 \rangle - | 1 \rangle))
    \label{equation:meyer1}& = | x \rangle \otimes \frac{1}{\sqrt{2}}(| 0 \oplus h_a(x) \rangle - | 1 \oplus h_a(x) \rangle)\\
    \label{equation:meyer2}& = (-1)^{h_a(x)} | x \rangle \otimes \frac{1}{\sqrt{2}}(| 0 \rangle - | 1 \rangle)\\ 
    \label{equation:meyer3}
    & = (-1)^{(n - \sum_{i = 1}^n \delta_{s_ix_i}) \bmod 2} | x \rangle \otimes \frac{1}{\sqrt{2}}(| 0 \rangle - | 1 \rangle)\\ 
    \label{equation:meyer3.5}
    & = (-1)^{n - \sum_{i = 1}^n \delta_{s_ix_i}} | x \rangle \otimes \frac{1}{\sqrt{2}}(| 0 \rangle - | 1 \rangle)\\ 
    \label{equation:meyer4}
    & = (-1)^n\mathop{\bigotimes}_{i = 1}^{n} (-1)^{\delta_{s_ix_i}}|x_i\rangle \otimes \frac{1}{\sqrt{2}}(| 0 \rangle - | 1 \rangle) 
\end{align}

\section{ $O(k\log k)$  Quantum Algorithm} \label{appendix:B}
Here we present a quantum algorithm with $O(k\log k)$ black-peg queries.
First, two functions will be employed, as described below: 
\begin{itemize}
    \item   $IPK_s$, associated with  $s \in [k]^n$, is defined by   $IPK_s(x)= \sum_{i} s_i \cdot x_i \bmod k$ for any  $x \in [k]^n$. 
    \item   $IPT_s$, associated with  $s \in [k]^n$,  is defined by  $IPT_s(x)= \sum_{i} s_i \cdot x_i \bmod k$ for any  $x \in [2]^n$.
\end{itemize}

\begin{algorithm}[htp]
    \caption{A  quantum algorithm for Mastermind  with $n$ positions and $ k$ colors}
    \label{algorithm:test}\label{nonadaptive}
    \LinesNumbered
    \SetKwInOut{KWProcedure}{Procedure}
    \SetKwInput{Runtime}{Runtime}
    \KwIn {A black-peg oracle $B_s$ for $s\in [k]^n$ such that $B_s| x \rangle |b\rangle = |x\rangle |b\oplus_{n+1} B_s(x)\rangle$. }
    \KwOut {The secret string $s$.}
    \Runtime{$O(k \log k)$  queries to  $B_s$. Succeeds with certainty.}
     \KWProcedure{}

      Prepare the initial  state $\ket{\Phi_0}=\ket{0}^{\otimes n}\ket{k-1}$, where $\ket{0}$ and $\ket{k-1}$ are basis states in a $k$-dimensional Hilbert space.
    
    Apply quantum Fourier transform $QFT_k ^{\otimes n+1}$.
    
    Apply the quantum  oracle of $IPK_s$  that calls the black-peg oracle $B_s$ $O(k \log k)$ times in parallel.
    
    Apply  inverse  quantum  Fourier transform $(QFT_k^\dagger) ^{\otimes n+1}$.
    
    Measure the first $n$ registers in the computational basis.
\end{algorithm}

Now one of our main results is the following theorem.

\begin{theorem}\label{theorem:Nonadapt_klogk}
There is a  quantum algorithm for the  Mastermind game with $n$ positions and $ k$ colors that uses $O(k \log k)$ black-peg queries and returns the secret string with certainty.
\end{theorem}

\begin{proof}
Our algorithm is described in Algorithm \ref{nonadaptive}. The main idea is to use a generalized version of the Bernstein-Vazirani algorithm \cite{Bernstein1997} by calling  $IPK_s$ one time to find $s$ and  we further show that  $IPK_s$ can be constructed by calling the black-peg function $B_s$ $O(k \log k)$ times in parallel. The process of Algorithm \ref{nonadaptive} can be depicted in Figure \ref{fig:nonadapt1}. Now assume that  $IPK_s$ is accessible. The state in Algorithm \ref{nonadaptive} evolves as follows.

\begin{figure}[htbp]  
\centering\includegraphics[width=12cm]{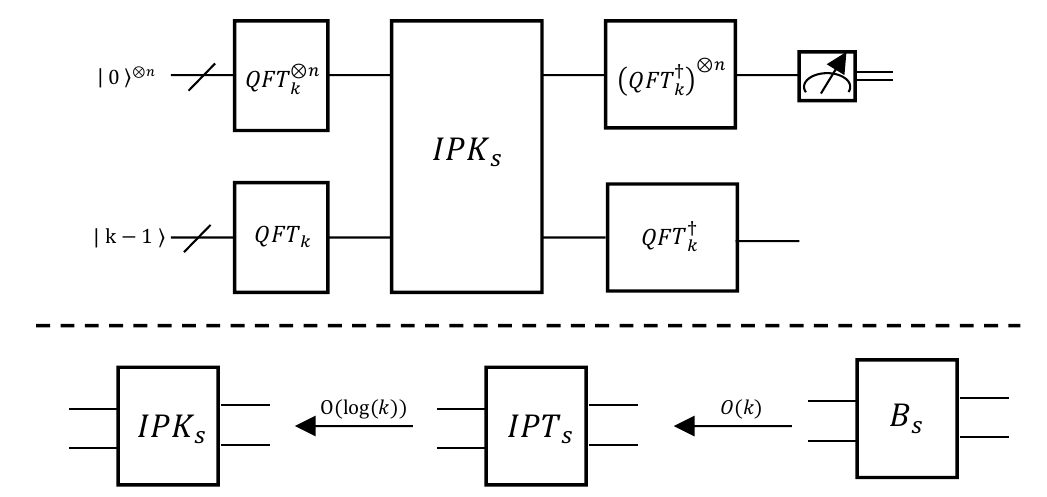} 
\caption{The circuit diagram of Algorithm \ref{nonadaptive} is depicted above the dashed line. The idea of how to construct $IPK_s$ is shown below the dashed line, where $A \stackrel{O(t)}{\longleftarrow} B$ means that  $A$ can be implemented by $O(t)$ copies of  $B$.}       
\label{fig:nonadapt1}   
\end{figure}

First, we prepare the initial state $\ket{\Phi_0}=\ket{0}^{\otimes n}\ket{k-1}$, where there are $n+1$ registers and each one is associated with a $k$-dimensional Hilbert space.

Second, after  quantum Fourier transform $QFT_k ^{\otimes n+1}$, the initial  state  is changed to
\begin{equation*}
    \ket{\Phi_1} = \frac{1}{\sqrt{k^n}} \sum _{x\in [k]^n } \ket{x}\ket{\phi},
\end{equation*}
where 
\begin{equation*}
    \ket{\phi} = \frac{1}{\sqrt{k}} \sum _{j=0}^{k-1} \omega^{k-j}\ket{j}
\end{equation*}
with $\omega=e^{2\pi i /k}$.

At the third step, apply  the quantum oracle of  $IPK_s$
\begin{equation*}
    IPK_s\ket{x}\ket{y} \longrightarrow \ket{x}\ket{(IPK_s(x) + y) \bmod k}. 
\end{equation*}
Then, the state evolves to
\begin{align*}
    \ket{\Phi_2}&=IPK_s\ket{\Phi_1}\\
    &= \frac{1}{\sqrt{k^n}} \sum _{x\in [k]^n } \omega^{IPK_s(x)} \ket{x}\ket{\phi}~~~~(\text {by Lemma}~ \ref{IPKS})\\
    &=\frac{1}{\sqrt{k^n}} \sum _{x\in [k]^n } e^{\frac{2\pi i \left(\sum_{j=1}^{j=n} s_j \cdot x_j\right) \bmod k}{k}}\ket{x}\ket{\phi}\\
    &= \frac{1}{\sqrt{k^n}} \sum _{x\in [k]^n } e^{\frac{2\pi i \sum_{j=1}^{j=n} s_j \cdot x_j }{k}}\ket{x}\ket{\phi}\\
    &= \frac{1}{\sqrt{k^n}} \sum _{x=x_1\dots x_n\in [k]^n } \prod_{j=1}^{j=n} e^{\frac{2\pi i s_j \cdot x_j }{k}}\ket{x_1\dots x_n}\ket{\phi}\\
     &= \frac{1}{\sqrt{k}}\sum_{x_1=0}^{k-1} e^{\frac{2\pi i s_1 \cdot x_1 }{k}}\ket{x_1} \otimes \dots \otimes \frac{1}{\sqrt{k}} \sum_{x_n=0}^{k-1} e^{\frac{2\pi i s_n \cdot  x_n }{k}}\ket{x_n}\ket{\phi}.
\end{align*}

At the fourth step, after applying inverse quantum Fourier transform $(QFT^\dagger) ^{\otimes n+1}$, we get the state 
\begin{equation*}
    \ket{\Phi_3}=\ket{s_1 s_2\dots s_n}\ket{k-1}.
\end{equation*}
Finally, the secret string $s=s_1s_2\dots s_n$ can be obtained with certainty after measuring the  first $n$ registers in the computational basis.

In the above procedure,  the $IPK_s$ oracle  is queried once and  can be constructed with $O(k \log k)$ queries to the black-peg function $B_s$ based on Lemma \ref{lemmaIPT}
and Lemma \ref{lemmaIPK}. Hence, the complexity of Algorithm \ref{nonadaptive}  with respective to $B_s$   is $O(k \log k)$. \end{proof}

\begin{lemma}\label{IPKS} Let   $IPK_s\ket{x}\ket{y} \longrightarrow \ket{x}\ket{(IPK_s(x) + y) \bmod k}$.  Then, for  $$\ket{\phi} = \frac{1}{\sqrt{k}} \sum _{j=0}^{k-1} \omega^{k-j}\ket{j}$$  with $\omega=e^{2\pi i /k}$, we have
\begin{equation*}
    IPK_s\ket{x}\ket{\phi} = \omega ^{IPK_s(x)}\ket{x}\ket{\phi}. 
\end{equation*}
\end{lemma}
\begin{proof}
Let $m_j=(IPK_s(x)+j) \bmod k$ for $j=0, 1,\cdots, k-1$. Then $IPK_s(x)+j-m_j=t_jk$ for some integer $t_j$, that is, 
\begin{equation*}
    j=t_jk+m_j-IPK_s(x).
\end{equation*}
Then we have 
\begin{align}
IPK_s\ket{x}\ket{\phi}&=\frac{1}{\sqrt{k}} \sum _{j=0}^{k-1} \omega^{k-j}IPK_s\ket{x}\ket{j}\\
&=\frac{1}{\sqrt{k}} \sum _{j=0}^{k-1} \omega^{k-j}\ket{x}\ket{(IPK_s(x)+j) \bmod k}\\
&=\sum _{j=0}^{k-1} \omega^{k-(t_jk+m_j-IPK_s(x))}\ket{x}\ket{m_j} \label{eq1}\\
&=\omega^{IPK_s(x)}\ket{x}\sum _{j=0}^{k-1} \omega^{k-m_j}\ket{m_j} \label{eqj}\\
&=\omega^{IPK_s(x)}\ket{x}\sum _{l=0}^{k-1} \omega^{k-l}\ket{l}\\
&=\omega^{IPK_s(x)}\ket{x}\ket{\phi},
\end{align}
where note that in Eq. (\ref{eq1}),   $\omega^{t_jk}=1$ holds for integer $t_j$, and in Eq. (\ref{eqj}), when $j$  traverses all the values in $\{0, 1, \cdots, k-1\}$, so does $m_j$.

\end{proof}

\begin{lemma}\label{lemmaIPK}
Given $s,x \in [k]^n$, $IPK_s(x)$ can be computed by calling  $IPT_s$  $ \lceil \log (k) \rceil $ times in parallel. 
\end{lemma}
\begin{proof}
 Given  $s = s_1s_2 \dots s_n \in [k]^n$ and $x = x_1x_2 \dots x_n \in [k]^n$, let $m= \lceil \log(k) \rceil$. $x_i(1) x_i(2) ... x_{i}(n)$ denotes the binary representation of $x_i$. There is   $x_i = \sum_{j=1}^{m} 2^{j-1} x_i(j)$ with $x_i(j) \in \{0, 1\}$. Then
    \begin{equation*}
    \begin{split}
      IPK_s(x)&=  \sum_{i = 1}^{n } s_i \cdot x_i \bmod k\\
        & = \sum_{i = 1} ^{n } \sum_{j = 1} ^{m } 2^{j-1} x_i(j) \cdot s_i \bmod k  \\
        & = \sum_{j = 1} ^{m } 2^{j-1} (\sum_{i = 1} ^{n}  x_i(j) \cdot s_i) \bmod k  \\
        & = \sum_{j = 1}^{m } 2^{j-1} ( \sum_{i = 1} ^{n}  x_i(j) \cdot s_i \bmod k) \bmod k\\
         & = \sum_{j = 1}^{m } 2^{j-1} IPT_s(x(j) ) \bmod k,
    \end{split}
    \end{equation*}
 where $x(j) = x_1(j) x_2(j) ... x_{n}(j)$. 
    
\end{proof}

\begin{lemma}\label{lemmaIPT}
Given $s\in [k]^n$ and $x\in [2]^n$, $IPT_s(x)$ can be computed by using $k$  black-peg queries $B_s$ in parallel.  
\end{lemma}
\begin{proof}
We now describe how to compute $IPT_s(x)=\sum_{i} s_i \cdot x_i \bmod k $ using  black-peg  queries $B_s$.

Given $x = x_1x_2 \dots x_n \in [2]^n$, we define $k$  strings $y^{c} \in [k]^n$ for $c = 0, 1, ..., k - 1$ as follows:
    \begin{equation*}
        y_i^{c}= 
        \begin{cases}
            & c, ~~~x_i=1, \\
            & 0, ~~~x_i=0.
        \end{cases}
    \end{equation*}
Feed the black-peg function $B_s$ with $y^{c}$,  and record the results as
    \begin{equation*}
        n_c = B_s(y^{c}) = |\{ i \in \{1,2,\dots, n\}: s_i = y_i^{c} \}|.
    \end{equation*}
Let 
    \begin{align*}
    &V = |\{ i \in \{1, 2, \dots, n\}: x_i = 1\}|,\\
         &v_c  = |\{i | s_i = c, i \in V\}|,\ c = 0, 1, ..., k - 1, \\
            &a  = |\{i | s_i = 0, i \in \{1,2,\dots, n\} - V \}|. 
    \end{align*}
For $c = 0, 1, ..., k - 1$,  obviously there are 
\begin{align}
        &n_c  = v_c + a, \label{njvja}\\
         &\sum_{c \in [k]} v_c  = |V|.\label{vjV}
\end{align}
Combine Eq.~\eqref{njvja} and Eq.~\eqref{vjV}, we have 
\begin{equation*}
    \sum_{c \in [k]} n_c  = |V| + k \cdot a.
\end{equation*}
Hence, we get 
\begin{equation*}
v_c  = n_c - \frac{\sum_{c \in [k]} n_c - |V| }{k}  
\end{equation*}
for $c = 0, 1, ..., k - 1$.
That is, we can use the query results $n_c$ to compute $v_c$ for $c = 0, 1, ..., k - 1$. Now we are ready to compute $IPT_s(x)$:
    \begin{equation*}
        IPT_s(x)= \sum_{i \in V} s_i \bmod k \\
        = \sum_{c \in [k]} c\cdot v_j\bmod k.
    \end{equation*}

As a result,  $IPT_s(x)$ can be computed by $k$ black-peg queries in parallel.
\end{proof}
\end{document}